\documentclass[11pt]{extarticle}
\def\version{
February 23, 2018
}

\nonstopmode

\usepackage[dvips]{graphics}
\usepackage{cancel} %to-be-commented-out
\usepackage{ulem} %to-be-commented-out

\providecommand{\emph}[1]{{\it #1}}
\renewcommand{\emph}[1]{{\it #1}}%% or else ulem changes \em to \underline

\usepackage[usenames,dvipsnames]{color}
\usepackage{amsthm}
\usepackage{latexsym,epsfig,bm}

\usepackage{bm}
\usepackage{upgreek}

\usepackage{mathrsfs}
\usepackage{times}
\usepackage{amssymb}

\definecolor{MyDarkBlue}{rgb}{0,0.08,0.45}
\definecolor{Pomegranade}{rgb}{0.6,0.1,0.15}
\definecolor{purple}{rgb}{0.6,0.1,0.15}
\usepackage[backref=none]{hyperref}
\hypersetup{pdfborder={0 0 0},
  colorlinks,
  urlcolor={MyDarkBlue},
  linkcolor={MyDarkBlue},
  citecolor={MyDarkBlue},
  breaklinks=true}
\providecommand{\url}[1]{\small\textcolor{blue}{#1}}
\usepackage{doi}
\providecommand{\eprint}[1]{}
\renewcommand{\eprint}[1]{arXiv:\href{http://arxiv.org/abs/#1}{#1}}

\usepackage{amsmath}

\providecommand{\eqref}[1]{{\rm (\ref{#1})}}

\DeclareSymbolFont{AMSb}{U}{msb}{m}{n}
\DeclareSymbolFontAlphabet{\mathbb}{AMSb}

\DeclareSymbolFont{EUR}{U}{eur}{m}{n}
\SetSymbolFont{EUR}{bold}{U}{eur}{b}{n}
\DeclareSymbolFontAlphabet{\eur}{EUR}

\DeclareSymbolFont{EUB}{U}{eur}{b}{n}
\SetSymbolFont{EUB}{bold}{U}{eur}{b}{n}
\DeclareSymbolFontAlphabet{\eub}{EUB}

\newcommand{\jj}{\mathrm{i}}
\newcommand{\End}{\,{\rm End}\,}

\newcommand{\bmK}{\bm{K}}

\newcommand{\e}{\bm{e}}

\newcommand{\notyet}[1]{{}}

\newcommand{\p}{\partial}
\newcommand{\at}[1]{\vert\sb{\sb{#1}}}

\def\R{\mathbb{R}}
\newcommand{\C}{\mathbb{C}}

\newcommand{\N}{\mathbb{N}}

\newcommand{\abs}[1]{\vert #1 \vert}

\newcommand{\norm}[1]{\Vert #1 \Vert}

\newcommand{\sothat}{\,\,{\rm ;}\ \ }

\newcommand{\ac}[1]{\noindent\textcolor{red}
{{\rm [\![}\mbox{\sc{AC}$\blacktriangleright\!\!\blacktriangleright$}: { #1}{\rm ]\!]}}}

\DeclareMathSymbol{\varGamma}{\mathord}{letters}{"00}
\DeclareMathSymbol{\varDelta}{\mathord}{letters}{"01}
\DeclareMathSymbol{\varTheta}{\mathord}{letters}{"02}
\DeclareMathSymbol{\varLambda}{\mathord}{letters}{"03}
\DeclareMathSymbol{\varXi}{\mathord}{letters}{"04}
\DeclareMathSymbol{\varPi}{\mathord}{letters}{"05}
\DeclareMathSymbol{\varSigma}{\mathord}{letters}{"06}
\DeclareMathSymbol{\varUpsilon}{\mathord}{letters}{"07}
\DeclareMathSymbol{\varPhi}{\mathord}{letters}{"08}
\DeclareMathSymbol{\varPsi}{\mathord}{letters}{"09}
\DeclareMathSymbol{\varOmega}{\mathord}{letters}{"0A}

\theoremstyle{plain}
\newtheorem{lemma}{Lemma}[section]
\newtheorem{theorem}{Theorem}[section]
\newtheorem{corollary}[lemma]{Corollary}
\newtheorem{proposition}[lemma]{Proposition}

\theoremstyle{definition}
\newtheorem{definition}[lemma]{Definition}
\theoremstyle{remark}
\newtheorem{remark}[lemma]{Remark}
\newtheorem{example}[lemma]{Example}

\newcounter{step}

\makeatletter\@addtoreset{equation}{section}
\makeatletter\@addtoreset{lemma}{section}
\makeatletter\@addtoreset{theorem}{section}
\makeatother

\newcommand{\Eta}{\mathrm{H}}
\def\Chi{X}

\renewcommand{\Re}{\mathop{\rm{R\hskip -1pt e}}\nolimits}
\renewcommand{\Im}{\mathop{\rm{I\hskip -1pt m}}\nolimits}

\begin{document}

\title{Spectral stability of
bi-frequency solitary waves
in Soler and Dirac--Klein--Gordon models}

\author{
{\sc Nabile Boussa{\"\i}d}
\\
{\it\small Universit\'e de Franche-Comt\'e, 25030 Besan\c{c}on CEDEX, France}
\\~\\
{\sc Andrew Comech}
\\
{\it\small Texas A\&M University, College Station, Texas 77843, USA}
\\%{\small and}
{\it\small IITP, Moscow 127051, Russia}
\\
{\it\small St. Petersburg State University, St. Petersburg 199178, Russia}
}

\date{\version}

\maketitle

\begin{abstract}
We construct bi-frequency solitary waves
of the nonlinear Dirac equation
with the scalar self-interaction,
known as the Soler model
(with an arbitrary nonlinearity and in arbitrary dimension)
and the Dirac--Klein--Gordon with Yukawa self-interaction.
These solitary waves provide a natural implementation
of qubit and qudit states in the theory of quantum computing.

%% We show that the linear stability
%% of these bi-frequency solitary waves
%% coincides with the linear stability
%% of the corresponding one-frequency solitary waves.
We show the relation of
$\pm 2\omega\mathrm{i}$ eigenvalues
of the linearization at a solitary wave,
Bogoliubov $\mathbf{SU}(1,1)$ symmetry,
and the existence of bi-frequency solitary waves.
We show that the spectral stability of these waves
reduces to spectral stability of usual (one-frequency) solitary waves.
%% We then set up the framework needed to study
%% the asymptotic stability of ordinary solitary waves
%% as well as of bi-frequency solitary waves.
\end{abstract}

\bigskip

\hfill
{\it To Vladimir Georgiev on the occasion of his 60th birthday}

\bigskip

\section{Introduction}
The Soler model \cite{jetp.8.260-short,PhysRevD.1.2766}
is the nonlinear Dirac equation
with the minimal scalar self-coupling,
\begin{eqnarray}\label{nld}
\jj\p\sb t\psi
=D_m\psi-f(\psi\sp\ast\beta\psi)\beta\psi,
\qquad
x\in\R^n,
\qquad
\psi(t,x)\in\C^N,
\end{eqnarray}
where $f$ is a continuous real-valued function with $f(0)=0$.
Above,
$\bar\psi=\psi\sp\ast\beta$,
with $\psi\sp\ast$ the hermitian conjugate.
This is one of the main models
of the nonlinear Dirac equation,
alongside with its own one-dimensional analogue,
the Gross--Neveu model
\cite{PhysRevD.10.3235,PhysRevD.12.3880},
and with the massive Thirring model \cite{MR0091788}.
Above, the Dirac operator is given by
\[
D_m=-\jj\bm\alpha\cdot\nabla+m\beta,
\qquad
m>0,
\]
with $\alpha^j$ ($1\le j\le n$) and $\beta$
mutually anticommuting
self-adjoint matrices
such that
$D_m^2=-\Delta+m^2$.
All these models are hamiltonian,
$\mathbf{U}(1)$-invariant, and relativistically invariant.
The classical field $\psi$ could be quantized
(see e.g. \cite{PhysRevD.12.3880}).

The Soler model shares the symmetry features
with its more physically relevant counterpart,
Dirac--Klein--Gordon system
(the Dirac equation with the Yukawa self-interaction,
which is also based on the quantity
$\bar\psi\psi$):
\begin{eqnarray}\label{dkg}
\jj\p_t\psi=D_m\psi-\Phi\beta\psi,
\qquad
(\p_t^2-\Delta+M^2)\Phi=\psi\sp\ast\beta\psi,
\qquad
x\in\R^n,
\end{eqnarray}
where
%$\psi(t,x)\in\C^N$
%and
$M>0$ is the mass of the scalar field $\Phi(t,x)\in\R$.

The solitary wave solutions
in the Soler model
(already constructed in \cite{PhysRevD.1.2766})
possess certain stability properties \cite{MR2892774,PhysRevLett.116.214101,linear-b};
in particular, small amplitude solitary waves
corresponding to the nonrelativistic limit $\omega\lesssim m$
of the charge-subcritical and charge-critical
case $f(\tau)=\abs{\tau}^k$ with $k\lesssim 2/n$
are \emph{spectrally stable}:
the linearized equation on the small perturbation
of a particular solitary wave
has no exponentially growing modes.
The opposite situation,
the \emph{linear instability} of small amplitude solitary waves
(presence of exponentially growing modes)
in the charge-supercritical case $k>2/n$ 
was considered in \cite{MR3208458}.
%% The (orbital) stability of the ground solitary waves in the NLS
%% is available if and only if the solitary waves can be
%% characterized as minima of the energy function
%% under the charge constraint \cite{MR901236},
%% while the absence of such a characterization
%% leads to linear instability \cite{VaKo};
%% on the other hand, the solitary waves in the nonlinear Dirac equation
%% do not have such a description and the stability 
%% might seem unlikely.
%% (Let us mention
%% that the absence of direct link between
%% stability and local energy minimization has already been
%% pointed out in \cite{MR722319}.)
Recent results on asymptotic stability
of solitary waves in the nonlinear Dirac equation
\cite{MR2924465,MR2985264,MR3592683}
rely on the assumptions on the spectrum of the
linearization at solitary waves,
although this information is not readily available,
especially in dimensions above one.
This stimulates the study of the spectra of linearizations
at solitary waves.
It was shown in \cite{MR2892774}
that the Soler model in one spatial dimension
linearized at a solitary wave
$\phi(x)e^{-\jj\omega t}$
has eigenvalues
$\pm 2\omega\jj$.
%%this was generalized to higher dimensions in \cite{dirac-vk}.
While the zero eigenvalues correspond to
symmetries of the system (unitary, translational, etc.),
the eigenvalues $\pm 2\omega\jj$
are related to
the presence of bi-frequency solitary waves
and to the Bogoliubov $\mathbf{SU}(1,1)$ symmetry
of the Soler model and Dirac--Klein--Gordon models,
first noticed by Galindo in \cite{MR0503135}.
In the three-dimensional case
($n=3$, $N=4$)
with the standard choice of the Dirac matrices,
this symmetry group takes the form
\begin{eqnarray}\label{def-g-b}
G\sb{\mathrm{Bogoliubov}}
=
\big\{a-\jj b\gamma^2\bmK\sothat
a,\,b\in\C,\ \abs{a}^2-\abs{b}^2=1\big\}\cong\mathbf{SU}(1,1),
\end{eqnarray}
where $\bmK:\C^N\to\C^N$ is the antilinear operator of complex conjugation;
the group isomorphism is given by
$a-\jj b\gamma^2\bmK\mapsto
\footnotesize\begin{bmatrix}a&b\\\bar b&\bar a\end{bmatrix}
$.
By the Noether theorem,
the continuous symmetry group leads to conservation laws
(see Section~\ref{sect-bogoliubov}).
The Bogoliubov group,
when applied to standard solitary waves $\phi(x)e^{-\jj\omega t}$
in the form of the Wakano Ansatz \cite{wakano-1966},
\begin{eqnarray}\label{wakano}
\phi(x)
=\begin{bmatrix}v(r,\omega)\xi\\\jj u(r,\omega)\frac{x\cdot\sigma}{\abs{x}}\xi\end{bmatrix},
\qquad
\xi\in\C^2,
\quad
\abs{\xi}=1,
\end{eqnarray}
produces bi-frequency solitary waves of the form
\begin{eqnarray}\label{qubits}
a\phi(x)e^{-\jj\omega t}+b\phi\sb{C}(x)e^{\jj\omega t},
\qquad
a,\,b\in\C,
\quad
\abs{a}^2-\abs{b}^2=1,
\end{eqnarray}
with $\phi\sb{C}=-\jj\gamma^2\bmK\phi$ the charge conjugate of $\phi$;
here
$-\jj\gamma^2\bmK$ is one of the infinitesimal generators of
$\mathbf{SU}(1,1)$.
Above, $v(r,\omega)$ and $u(r,\omega)$ are real-valued functions
which satisfy
\begin{eqnarray}\label{def-v-u}
\begin{cases}
\omega v=\p_r u+\frac{n-1}{r}u+(m-f)v,
\\
\omega u=-\p_r v-(m-f)u,
\end{cases}
\quad
\lim\sb{r\to 0}u(r,\omega)=0,
\end{eqnarray}
with 
\begin{eqnarray}\label{f-v-u}
f=f(v^2-u^2)
\end{eqnarray}
in the case of nonlinear Dirac equation \eqref{nld}
(see e.g.
%%\cite[Eq. (1.9)]{MR1344729} or
\cite{MR3670258})
and
\begin{eqnarray}\label{f-not-v-u}
f=(-\Delta+M^2)^{-1}(v^2-u^2)
\end{eqnarray}
in the case of Dirac--Klein--Gordon system \eqref{dkg}.
We assume that
the functions $u(r,\omega)$ and $v(r,\omega)$
satisfy
\begin{eqnarray}\label{u-less-v}
%%\sup\sb{r\ge 0}\Abs{\frac{u(r)}{v(r)}}<1;
\sup\sb{r\ge 0}\abs{u(r,\omega)/v(r,\omega)}<1,
\end{eqnarray}
which is true in particular for small amplitude solitary waves
with $\omega\lesssim m$;
see e.g. \cite{MR3670258}.

\bigskip

We now switch to the general case of a general spatial dimension $n\ge 1$
and a general number of spinor components $N\ge 2$.
By usual arguments
(see e.g. \cite{MR3670258}),
without loss of generality,
we may assume that the Dirac matrices have the form
\begin{eqnarray}\label{general-alpha}
\alpha^j
=\begin{bmatrix}
0&\sigma_j\sp\ast\\\sigma_j&0
\end{bmatrix},
\qquad
1\le j\le n;
\qquad
\beta=\begin{bmatrix}1\sb{N/2}&0\\0&-1\sb{N/2}\end{bmatrix}.
\end{eqnarray}
Here $\sigma\sb j$, $1\le j\le n$,
are
$\frac N 2\times\frac N 2$ matrices
which are the higher-dimensional analogue Pauli matrices:
\begin{eqnarray}\label{general-sigma}
\{\sigma\sb j,\sigma\sb k\sp\ast\}=2\delta\sb{j k}1\sb{N/2},
\qquad
1\le j,\,k\le n.
\end{eqnarray}
%\begin{remark}
%\end{remark}
%% The Dirac operator is then given by
%% $
%% D_m=-\jj
%% \begin{bmatrix}
%% 0&\sigma\sp\ast\cdot\nabla
%% \\
%% \sigma\cdot\nabla&0
%% \end{bmatrix}
%% +m
%% \begin{bmatrix}1\sb{N/2}&0\\0&-1\sb{N/2}\end{bmatrix}
%% .
%% $
In the case $n=3$, $N=4$,
%% we have
%% \begin{eqnarray}
%% \label{standard}
%% \alpha^j=
%% \begin{bmatrix}0&\sigma_j\\\sigma_j&0\end{bmatrix},
%% \quad
%% 1\le j\le 3,
%% \qquad
%% \beta=
%% \begin{bmatrix}I_{2}&0\\0&-I_{2}\end{bmatrix},
%% \end{eqnarray}
one takes $\sigma_j$, $1\le j\le 3$, to be the standard Pauli matrices.

\begin{remark}\label{remark-n-4}
In general, $\sigma_j$ are not necessarily self-adjoint;
for example,
for $n=4$ and $N=4$,
one can choose
$\sigma_j$ to be the standard Pauli matrices
for $1\le j\le 3$
and set
$\sigma_4=\jj I_2$.
\end{remark}

We denote
\begin{eqnarray}\label{def-sigma-r}
\sigma_r=\frac{x\cdot\sigma}{r},
\qquad
x\in\R^n\setminus\{0\},
\qquad
r=\abs{x}.
\end{eqnarray}
In Section~\ref{sect-bi-frequency},
we show that
if $\phi(x)e^{-\jj\omega t}$,
with $\phi$ from \eqref{wakano},
is a solitary wave solution
to \eqref{nld} (or \eqref{dkg}),
then
there is the following family of exact solutions
to \eqref{nld} (or \eqref{dkg}, respectively):
\begin{eqnarray}\label{qudits}
&\theta_{\Xi,\Eta}(t,x)
=\abs{\Xi}\phi\sb\xi(x)e^{-\jj\omega t}+\abs{\Eta}\chi\sb\eta(x)e^{\jj\omega t},
\\[1ex]
\nonumber
&
\phi\sb\xi(x)
=\begin{bmatrix}v(r,\omega)\xi\\\jj u(r,\omega)\sigma_r\xi\end{bmatrix},
\qquad
\chi\sb\eta(x)
=\begin{bmatrix}-\jj u(r,\omega)\sigma_r\sp\ast\eta\\v(r,\omega)\eta\end{bmatrix}
,
\end{eqnarray}
with
$
\Xi,\,\Eta\in\C^{N/2}\setminus\{0\},
$
$
\abs{\Xi}^2-\abs{\Eta}^2=1,
$
$
\xi=\Xi/\abs{\Xi},
$
$
\eta=\Eta/\abs{\Eta}$.
See Lemma~\ref{lemma-bi} below.
This shows that in any dimension
there is a larger symmetry group,
$\mathbf{SU}(N/2,N/2)$,
which is present at the level of bi-frequency solitary wave solutions
in the models \eqref{nld} and \eqref{dkg}
while being absent at the level of the Lagrangian.

\begin{remark}
We note that if $f$ in \eqref{nld} is even,
then
$\theta_{\Xi,\Eta}(t,x)$
given by \eqref{qudits}
with
$\Xi,\,\Eta\in\C^{N/2}$
such that
$\abs{\Xi}^2-\abs{\Eta}^2=-1$
is also a solitary wave solution.
\end{remark}

Two-frequency solitary waves \eqref{qudits}
clarify the nature of the eigenvalues $\pm 2\omega\jj$
of the linearization at (one-frequency) solitary waves
in the Soler model:
these eigenvalues could be interpreted
as corresponding to the tangent vectors
to the manifold of bi-frequency solitary waves.
See Corollary~\ref{corollary-bi} below.
We point out that the exact knowledge of the presence of $\pm 2\omega\jj$
eigenvalues
in the spectrum of the linearization at a solitary wave
is important for the proof of the spectral stability:
namely, it allows us to conclude that
in the nonrelativistic limit $\omega\lesssim m$
the only eigenvalues that bifurcate
from the embedded thresholds at $\pm 2m\jj$
are $\pm 2\omega\jj$;
no other eigenvalues can bifurcate from $\pm 2m\jj$,
and in particular no eigenvalues with nonzero real part.
For details, see \cite{linear-b}.

We point out that
the asymptotic stability of standard, \emph{one-frequency} solitary waves
can only hold
with respect to the whole manifold of solitary wave solutions
\eqref{qudits},
which includes both one-frequency and bi-frequency solitary waves:
if a small perturbation of a one-frequency solitary wave is
a bi-frequency solitary wave, which is an exact solution,
then convergence to the set of one-frequency solitary waves
is out of question.
In this regard,
we recall that the asymptotic stability results
\cite{MR2985264,MR3592683}
were obtained under certain restrictions
on the class of perturbations.
It turns out that these restrictions were
sufficient to remove not only translations,
but also the perturbations in the directions
of bi-frequency solitary waves;
this is exactly why the proof of asymptotic
stability
of the set of \emph{one-frequency} solitary waves
with respect to such class of perturbations
was possible at all.

While the stability of one-frequency solitary waves
turns out to be related to the existence of bi-frequency solitary waves,
%%(for details, see \cite{linear-b}),
one could question the stability of such bi-frequency solutions, too.
In Section~\ref{sect-bi-frequency-stability}, we show that
the bi-frequency solitary waves are spectrally stable as long as
so are the corresponding one-frequency solitary waves.
While this conclusion may seem natural,
we can only give the proof for the case
when the number of spinor components satisfies $N\le 4$
(which restricts the spatial dimension to $n\le 4$).

\medskip

Let us mention that the bi-frequency solitary waves \eqref{qubits}
may play a role in Quantum Computing.
Indeed, such states produce a natural implementation of qubit states
$a\vert 0\rangle+b\vert 1\rangle$, 
$\abs{a}^2+\abs{b}^2=1$,
except that now the last relation takes the form
$\abs{a}^2-\abs{b}^2=1$.
Just like for standard cubits,
our bi-frequency states \eqref{qubits}
have two extra parameters besides the
orbit of the $\mathbf{U}(1)$-symmetry group.
Below, we are going to show that qubits \eqref{qubits}
can be linearly stable.
Moreover, the manifold of bi-frequency solitary waves
\eqref{qudits}
admits a symmetry group $\mathbf{SU}(N/2,N/2)$
(which may be absent on the level of the Lagrangian).
For these solitary waves,
the number of degrees of freedom
(after we factor out the action of the unitary group)
is
$d=2N-2$.
The states \eqref{qudits} with $N\ge 4$ correspond to
higher dimensional versions of qubits --
$d$-level \emph{qudits},
quantum objects for which the number of possible states is greater than two.
These systems could implement quantum computation
via compact higher-level quantum structures,
leading to novel algorithms in the theory of quantum computing.
Bi-frequency solitary waves could also provide a simple
stable realization
of higher-dimensional quantum entanglement, or
\emph{hyperentanglement},
which is used in cryptography based on quantum key distribution;
by \cite{cerf2002security,PhysRevA.69.032313},
using qudits over qubits
provides increased coding density for higher security margin
and also an increased level of tolerance to noise at a given level of security.
%that using qudits over qubits
%may  improve quantum cryptography
%based on quantum key distribution.
%% and a higher flux of information 
Qudits have already been implemented
in the system with two electrons \cite{melnikov2016}
as \emph{quantum walks of several electrons}
(just like one qubit could be represented by a \emph{quantum walk},
a distribution of an electron in a ``quantum tunnel''
between individual quantum dots, considered as potential wells).
Being sensitive to the external noise,
the quantum-walk implementation of qudits in \cite{melnikov2016}
was indicated to be highly unstable,
requiring excessive cooling and
making the practical usage very difficult.
In \cite{kues2017},
the on-chip implementation of qudit states
is achieved by creating
photons in a coherent superposition of multiple high-purity
frequency modes.
We point out that
the bi-frequency solitary waves \eqref{qudits}
can possess stability properties, as we show below;
moreover, the simplicity of the model
suggests that such states could be implemented using
photonic states in optical fibers
without excessive quantum circuit complexity.

We need to mention
%the optical analogues of the Dirac equation 
%\cite{Tran2014179,PhysRevA.79.053830}. 
that several novel nonlinear photonic systems currently 
explored are modeled by Dirac-like equations
(often called \emph{coupled mode systems})
which are similar to \eqref{nld}.
Examples include fiber Bragg gratings \cite{JOSAB.14.002980}, 
dual-core photonic crystal fibers \cite{OL.31.001480}, and discrete binary arrays,
which refer to systems built as arrays coupled of elements of two types.
Earlier experimental work on binary arrays has already
shown the formation of discrete gap solitons \cite{OL.29.002890}. 
Three of the many novel examples that have been recently considered are: 
a dielectric metallic waveguide array \cite{2013OptCo.297..125A,6646384}; an array of 
vertically 
displaced binary waveguide arrays with longitudinally modulated effective 
refractive index \cite{PhysRevLett.113.150401},
and arrays of coupled parity-time ($\mathcal{PT}$) nanoresonators  \cite{PhysRevLett.110.053901}. 
We also constructed bi-frequency solutions of the form \eqref{qubits}
in the Dirac-type models with the $\mathcal{PT}$-symmetry
which arise in nonlinear optics
in the model describing arrays of optical fibers
with gain-loss behavior \cite{JSTQE.2015.2485607}.
This venue of research is
pursued for optical implementation of traditional circuits
\cite{PhysRevA.84.040101,1751-8121-45-44-444029}
aimed at the energy-efficient computing
and at the challenges
in reducing the footprint of optics-based devices.

%% If all perturbations are allowed,
%% then the asymptotic stability could only be proved
%% for the convergence
%% to the entire orbit of translations
%% and the $\mathbf{SU}(1,1)$ transformations
%% through the original solitary wave
%% (the whole set of exact solutions near the unperturbed
%% solitary wave).
%% Since the invariance group turns out to be noncommutative,
%% the asymptotic stability is to be proved
%% in the context of the second paper
%% of the seminal sequence \cite{MR901236,MR1081647}.
%%
%% We will study the symmetry of the Soler model \cite{PhysRevD.1.2766},
%% which is a nonlinear Dirac equation with the scalar self-interaction.
%% As a by-product, we will construct bi-frequency solitary waves
%% in this model.
%%
%% Let us also mention that the massive Thirring model,
%% based on the vector self-interaction,
%% is completely integrable in one spatial dimension;
%% as a consequence,
%% the solitary waves possess orbital stability
%% \cite{MR3147657,MR3462129}.

\section{The Bogoliubov $\mathbf{SU}(1,1)$ symmetry
and associated charges}
\label{sect-bogoliubov}

Chadam and Glassey \cite{MR0369952}
noted an interesting feature of the model \eqref{dkg}:
under the standard choice of $4\times 4$ Dirac matrices
$\alpha^j$ and $\beta$,
as long as the solution $\psi$
is sufficiently regular,
there is a conservation of the quantity
\begin{eqnarray}\label{charge-cg}
\int\sb{\R^3}
\big(\abs{\psi_1-\bar\psi_4}^2
+
\abs{\psi_2+\bar\psi_3}^2
\big)\,dx.
\end{eqnarray}
As a consequence, if \eqref{charge-cg}
is zero at some and hence at all
moments of time,
%belongs to
%$C(\R,L^2(\R^n,\C^N))$ so that
%\eqref{charge-cg} is continuous in time
then
$\abs{\psi_1}=\abs{\psi_4}$
and $\abs{\psi_2}=\abs{\psi_3}$
for almost all $x$ and $t$,
hence $\psi\sp\ast\beta\psi\equiv 0$
(in the distributional sense),
%% \ac{better to provide the LWP references},
meaning that the self-interaction plays no role
in the evolution and that as the matter of fact
the solution solves the linear equation
(without self-interaction).
A similar feature of the Soler model \eqref{nld}
was analyzed in \cite{MR2082456}.
%% and we define
%% $
%% g(\tau)=m-f(\tau),
%% $
%% so that we could write \eqref{nld} as
%% \begin{equation}\label{nld-g}
%% i\p_t\psi
%% =D_m\psi-f(\psi\sp\ast\beta\psi)\beta\psi.
%% %=D_0\psi+g(\psi\sp\ast\beta\psi)\beta\psi.
%% \end{equation}
The relation of the conservation
of the quantity \eqref{charge-cg}
in Dirac--Klein--Gordon
to the $\mathbf{SU}(1,1)$ symmetry
of the corresponding Lagrangian
based on combinations of $\bar\psi \gamma\sp 0 D_m\psi$
and $\bar\psi\psi$
was noticed by Galindo \cite{MR0503135}.

Let us state the above results in a slightly more general setting.
%%of \eqref{dkg} and \eqref{nld}.
Assume that $B\in\End(\C^N)$
is a matrix which satisfies
\begin{eqnarray}\label{b-such}
\{B\bmK,D_m\}=0,
\quad
B\bmK=\bmK B\sp\ast,
\quad
B\sp\ast B=1\sb{N}.
\end{eqnarray}
%% \begin{eqnarray}\label{b-such-1}
%% &&
%% \{B\bmK,D_m\}=0,
%% \\[1ex]
%% \label{b-such-2}
%% &&
%% B\sp\ast B=1\sb{N},
%% \qquad
%% B\bmK=\bmK B\sp\ast.
%% \end{eqnarray}
The relations \eqref{b-such} imply that
\begin{eqnarray}
&
(B\bmK)^2
=\bmK B\sp\ast B\bmK=1\sb{N},
\nonumber
\\
&
\label{b-such-2a}
(\beta B)\sp{\mathrm{t}}
=\bmK(\beta B)\sp\ast\bmK
=\bmK B\sp\ast \beta\bmK
=-\beta\bmK B\sp\ast\bmK
=-\beta B.
\end{eqnarray}
Above, ``$\mathrm{t}$'' denotes the transpose.
%Their result \cite[Lemma 1]{MR2082456}
%could be stated as follows.

\begin{remark}
One can think of $\bmK$ as self-adjoint in the sense
that
\[
\Re[(\bmK\psi)\sp\ast\theta]
=
\Re[\overline{\psi\sp\ast\bmK\theta}]
=
\Re[\psi\sp\ast\bmK\theta].
\]
\end{remark}

%% \begin{remark}
%% For fixed $n$ and $N$ large enough,
%% there could be several possible choices of $B$
%% satisfying \eqref{b-such},
%% corresponding to a larger symmetry group.
%% \end{remark}

%% According to \cite{MR570119},
%% {\it
%% ...It is shown that infinitesimal transformations of a Dirac spinor,
%% in which small amounts of negative energy states are mixed with
%% positive energy states, together with infinitesimal phase transformation,
%% form the generators of an $\mathrm{sl}(2,\R)$-algebra.
%% The global invariance of the Dirac equation obtained in this way
%% is extended to a local invariance by introducing, in addition to
%% the electromagnetic potential, another complex potential,
%% which carries negative energy and is doubly charged.
%% The field equations for the Dirac field coupled to the new gauge fields
%% are derived and a number of special solutions are given.
%% }
%% \end{remark}

We can summarize \cite{MR0369952,MR2082456}
as follows:

\begin{lemma}
If the solution to
\eqref{nld} or \eqref{dkg}
satisfies
$B\bmK\psi\at{t=0}= z\psi\at{t=0}$, for some $z\in\C$,
$\abs{z}=1$,
then
$B\bmK\psi=z\psi$ for all $t\in\R$
and moreover
$\psi\sp\ast\beta\psi=0$.
\end{lemma}

\begin{proof}
The first statement is an immediate consequence of the fact that
$B\bmK$ commutes with the flow of the equation.
If $\psi(t,x)$ satisfies
$\jj\dot\psi=D_m\psi-f(\psi\sp\ast\beta\psi)\beta\psi$,
then
\[
B\bmK\dot\psi
=B\bmK(-\jj(D_m-\beta f))\psi
=(-\jj(D_m-\beta f))B\bmK\psi
=z(-\jj D_m+\jj\beta f)\psi
=z\dot\psi.
\]
Finally, since $B\bmK\psi=z\psi$,
%%and $(\beta B)\sp t=-\beta B$,
then
\[
z\psi\sp\ast\beta\psi
=\psi\sp\ast\beta B\bmK\psi
=(\bmK\psi)\sp t\beta B\bmK\psi
=-(\bmK\psi)\sp t\beta B\bmK\psi
=0.
\]
We took into account that
$(\beta B)\sp t=-\beta B$
by \eqref{b-such-2a}.
%% Similarly,
%% \[
%% \psi\sp\ast\gamma^5\beta\psi
%% =z^{-1}(\bmK\psi)\sp t\gamma^5\beta B\bmK\psi
%% =-\jj z^{-1}(\bmK\psi)\sp t\gamma^5\alpha^2\bmK\psi=0
%% \]
%% and
%% \[
%% \psi\sp\ast\gamma^5\psi
%% =z^{-1}(\bmK\psi)\sp t\gamma^5 B\bmK\psi
%% =-\jj z^{-1}(\bmK\psi)\sp t\gamma^5\gamma^2\bmK\psi=0.
%% \]
\end{proof}

As in \cite{MR0503135},
the Lagrangians
of the Soler model \eqref{nld}
and
of the Dirac--Klein--Gordon model \eqref{dkg},
with the densities
\[
\mathscr{L}\sb{\mathrm{Soler}}
=\psi\sp\ast D_m\psi+F(\bar\psi\psi),
\]
\[
\mathscr{L}\sb{\mathrm{DKG}}
=\psi\sp\ast D_m\psi+\Phi\bar\psi\psi
+\frac{1}{2}
\big(
\abs{\dot\Phi}^2+\abs{\nabla\Phi}^2+M^2\abs{\Phi}^2
\big),
\]
with $\psi(t,x)\in\C^N$, $\Phi(t,x)\in\R$,
are invariant under the action of the
continuous symmetry group
\[
g\in G\sb{\mathrm{Bogoliubov}},
\qquad
g:\;
\psi\mapsto
(a+b B\bmK)\psi,
\qquad
\abs{a}^2-\abs{b}^2=1
\]
(cf. \eqref{def-g-b}).
The Noether theorem
leads to the conservation of the standard charge
$Q
=\int\sb{\R^3}\psi\sp\ast\psi\,dx$
corresponding to the standard charge-current density
$\bar\psi\gamma\sp\mu\psi$
(note that
the unitary group is a subgroup of $\mathbf{SU}(1,1)$),
and the complex-valued Bogoliubov charge
$
\Lambda
=\int\sb{\R^n}\psi\sp\ast B\bmK\psi\,dx
$
which corresponds to the complex-valued four-current density
$\psi\sp\ast\gamma^0\gamma\sp\mu B\bmK\psi$.
Now Galindo's observation \cite{MR0503135}
could be stated as follows.

\begin{lemma}[The Bogoliubov $\mathbf{SU}(1,1)$ symmetry
and the charge conservation]~
\label{lemma-ab}

\begin{enumerate}
\item
If $\psi(t,x)\in\C^N$ is a solution to
\eqref{nld}
or
\eqref{dkg},
then so is
$g\psi(t,x)$,
for any $g\in G\sb{\mathrm{Bogoliubov}}$.
\item
The Hamiltonian density
of the Soler model \eqref{nld},
\begin{equation}\label{hamiltonian}
\mathscr{H}=\psi\sp\ast D\sb m\psi-F(\psi\sp\ast\beta\psi),
\end{equation}
where
$
F(\tau)=\int_0^\tau f(t)\,dt$,
$\tau\in\R$,
satisfies
%is invariant under the action of $G\sb{\mathrm{Bogoliubov}}$,
%so that
$
\mathscr{H}(g\psi)
=
\mathscr{H}(\psi)
$,
$
\forall
g\in G\sb{\mathrm{Bogoliubov}}.
$
\item
$G\sb{\mathrm{Bogoliubov}}
\cong \mathbf{SU}(1,1)$,
with the group isomorphism
$
a+b B\bmK
\mapsto
\begin{bmatrix}a&b\\\bar b&\bar a\end{bmatrix}
\in\mathbf{SU}(1,1),
$
where $a,\,b\in\C$ satisfy $\abs{a}^2-\abs{b}^2=1$.
\item
For solutions of the nonlinear Dirac equation \eqref{nld},
the following quantities are (formally) conserved:
\begin{eqnarray}
\nonumber
&&
Q(\psi)
=\big\langle\psi,\psi\big\rangle
=\int\sb{\R^n}\psi(t,x)\sp\ast\psi(t,x)\,dx,
\\
&&
\nonumber
\Lambda(\psi)
=\big\langle \psi,B\bmK\psi\big\rangle
=\int\sb{\R^n}
B\sb{j k}
\overline{\psi\sb j(t,x)}
\,
\overline{\psi\sb k(t,x)}
\,dx.
\end{eqnarray}
\end{enumerate}
\end{lemma}

\begin{proof}
Let $g=a+b B\bmK$,
$\abs{a}^2-\abs{b}^2=1$.
Since
$B\bmK$ anti-commutes with
both $\jj$ and with $D_m$,
one has
$
\jj\p\sb t(B\bmK\psi)
=
(D_m-f(\psi\sp\ast\beta\psi)\beta)(B\bmK\psi).
$
Taking the linear combination with \eqref{nld},
we arrive at
\[
\jj\p\sb t(a+b B\bmK)\psi
=
(D_m-f(\psi\sp\ast\beta\psi)\beta)((a+b B\bmK)\psi).
\]
It remains to notice that
$\overline{\varphi\sp\ast\bmK \rho}=(\bmK \varphi)\sp\ast \rho$,
$\forall\varphi,\,\rho\in\C^N$,
hence
\[
\Re\{\varphi\sp\ast\bmK \rho\}=\Re\{(\bmK\varphi)\sp\ast \rho\},
\]
resulting in
\begin{eqnarray}\label{a-b-beta-a-b}
(g\psi)\sp\ast\beta g\psi
=\Re\{\psi\sp\ast(\bar a+\bmK B\sp\ast \bar b)\beta(a+b B\bmK)\psi\}
\nonumber
\\
=\Re\{\psi\sp\ast(\bar a+b B\bmK)(a-b B\bmK)\beta\psi\}
=\psi\sp\ast\beta\psi.
\end{eqnarray}
The invariance of the Hamiltonian density follows from
(cf. \eqref{a-b-beta-a-b})
\[
(g\psi)\sp\ast D\sb m g\psi
=\Re\{\psi\sp\ast(\bar a+b B\bmK)(a-b B\bmK)D\sb m
\psi\}
=\psi\sp\ast D\sb m\psi
\]
and from
$F((g\psi)\sp\ast\beta g\psi)=F(\psi\sp\ast\beta\psi)$.

By the N\"other theorem,
the invariance under the action
of a continuous group
results in the conservation laws.
%\begin{proof}
Let us check the (formal) conservation
of the complex-valued $\Lambda$-charge.
Writing
$f=f(\psi\sp\ast\beta\psi)$, we have:
\[
\p\sb t \Lambda(\psi)
=
\big\langle -\jj(D_m-f\beta)\psi,B\bmK\psi\big\rangle
+
\big\langle\psi,B\bmK(-\jj(D_m-f\beta)\psi)\big\rangle
\]
\[
=
\big\langle\psi,\jj(D_m-f\beta)B\bmK\psi\big\rangle
+
\big\langle\psi,\jj B\bmK(D_m-f\beta)\psi\big\rangle
=0.
\]
In the last relation,
we took into account the anticommutation relations
from \eqref{b-such}.
We also note that for the densities,
we have
\[
\p\sb t
(\psi\sp\ast B\bmK\psi)
=
-(\alpha\sp j\p\sb j\psi)\sp\ast B\bmK\psi
-\psi\sp\ast B\bmK\alpha\sp j \p\sb j\psi
=
-\p\sb{x\sp j}
(\psi\sp\ast\alpha\sp j\psi B\bmK\psi),
\]
showing that the Minkowski vector
of the Bogoliubov charge-current density is given by
\[
S\sp\mu(t,x)
=\psi(t,x)\sp\ast\gamma\sp 0\gamma\sp\mu B\bmK\psi(t,x).
\qedhere
\]
\end{proof}

\begin{remark}
Three conserved quantities,
one being real and one complex,
correspond to $\dim\sb{\R}\mathbf{SU}(1,1)=3$.
\end{remark}

\begin{example}
For $N=2$
and $n \le 2$,
with $D_m=-\jj\sum_{j=1}\sp{n}\sigma_j\p_j+\sigma_3 m$,
one takes $B=\sigma_1$ ($\sigma_j$ being the standard Pauli matrices);
%% If $n\le 2$,
%% then we use $\sigma\sb j$, $1\le j\le 2$,
%% for the spatial Dirac matrices,
%% and set $\beta=\sigma_3$.
%% One then takes $B=\sigma_1$,
\[
B\bmK\psi=\sigma_1\bmK\psi=:\psi\sb{C},
\qquad
\psi\in\C^2.
\]
The conserved quantity is
\[
\Lambda
=
\int\sb\R\psi\sp\ast\sigma_1\bmK\psi\,dx
=
2\int\sb\R\bar\psi_1\bar\psi_2\,dx.
\]
It follows that the charge can be decomposed into
\[
Q=Q\sb{-}+Q\sb{+},
\]
with both
\[
Q\sb\pm
:=\frac 1 2
(Q\pm\Re\Lambda)
=\frac 1 2\int\sb\R
(\abs{\psi_1}^2+\abs{\psi_2}^2
\pm 2\Re\psi_1\psi_2)\,dx
=\frac 1 2\int\sb\R
\abs{\psi_1\pm\bar\psi_2}^2\,dx
\]
conserved in time; so, if at $t=0$ one has
$\bar\psi_2=\psi_1$
(or, similarly, if $\bar\psi_2=-\psi_1$),
then this relation persists for all times
(hence
$\psi\sp\ast\sigma_3\psi=\abs{\psi_1}^2-\abs{\psi_2}^2=0$
for all times)
due to the conservation of $Q\sb\pm$.
\end{example}

\begin{example}
In the case $n=3$, $N=4$,
using the standard choice of the Dirac matrices,
one can take $B=-\jj\gamma^2$, so that
\[
B\bmK\psi=-\jj\gamma^2\bmK\psi=:\psi\sb{C},
\qquad
\psi\in\C^4.
\]
Then the quantity
\[
\Lambda(\psi)
=\int\sb{\R^3}\psi\sp\ast(-\jj\gamma\sp 2)\bmK\psi\,dx
=2\int\sb{\R^3}
(-\bar\psi_1\bar\psi_4+\bar\psi_2\bar\psi_3)
\,dx
\]
is conserved, hence so are
\[
Q\sb\pm
:=
\frac 1 2(Q\pm\Re\Lambda)
=
\frac 1 2
\int\sb{\R^3}
\big(
\abs{\psi_1\mp\bar\psi_4}^2
+
\abs{\psi_2\pm\bar\psi_3}^2
\big)\,dx.
\]
Thus, if at some moment of time
one has
$\bar\psi_4=\psi_1$ and $\bar\psi_3=-\psi_2$
(or, similarly, if
$\bar\psi_4=-\psi_1$ and $\bar\psi_3=\psi_2$),
then this relation persists for all times
(hence $\psi\sp\ast\beta\psi=0$)
due to the conservation of $Q\sb\pm$.
Note that $2Q\sb{+}$ coincides with \eqref{charge-cg},
so our conclusions are in agreement with
\cite{MR0369952,MR2082456}.
\end{example}

\begin{remark}
For $n=4$ and $N=4$ (cf. Remark~\ref{remark-n-4}),
there is no $B$ satisfying \eqref{b-such}
and thus no $\mathbf{SU}(1,1)$ symmetry.
\end{remark}

\begin{lemma}[Transformation of the charges
under the action of $\mathbf{SU}(1,1)$]
\label{lemma-q-lambda-transform}
Let $a,\,b\in\C$, $\abs{a}^2-\abs{b}^2=1$,
so that
$
g=a+b B\bmK\in G\sb{\mathrm{Bogoliubov}}.
$
Then
$g\psi=(a+b B\bmK)\psi$ satisfies
\begin{eqnarray}\label{new-q}
&&
Q((a+b B\bmK)\psi)
=(\abs{a}^2+\abs{b}^2)Q(\psi)+2\Re\{\bar a b \Lambda(\psi)\},
\\
&&
\label{new-lambda}
\Lambda(a+b B\bmK\psi)
=
\bar a^2 \Lambda(\psi)
+
2 \bar a \bar b Q(\psi)
+
\bar b^2\overline{\Lambda(\psi)}.
\end{eqnarray}
The quantity
$Q^2-\abs{\Lambda}^2$
is invariant
under the action of $\mathbf{SU}(1,1)$:
\begin{equation}\label{q-lambda-invariant}
Q(g\psi)^2-\abs{\Lambda(g\psi)}^2
=
Q(\psi)^2-\abs{\Lambda(\psi)}^2.
\end{equation}
\end{lemma}

\begin{proof}
For the charge density of $g\psi=(a+b B\bmK)\psi$, one has
\[
(g\psi)\sp\ast g\psi
=\Re\{\psi\sp\ast(\bar a+b B\bmK)(a+b B\bmK)\psi\}
=\Re\{\psi\sp\ast(\abs{a}^2+\abs{b}^2+2\bar a b B\bmK)\psi\}
\]
\[
=(\abs{a}^2+\abs{b}^2)\psi\sp\ast\psi
+2\Re\{\bar a b\psi\sp\ast B\bmK\psi\}.
\]
For the $\Lambda$-charge density of $g\psi$,
using the identity
%%$(\bmK u)\sp\ast v=\overline{u\sp\ast\bmK v}$,
$(B\bmK u)\sp\ast B\bmK v=\overline{u\sp\ast v}$,
$\forall u,\,v\in\C^N$
which follows from \eqref{b-such},
one has
\[
(g\psi)\sp\ast B\bmK g\psi
=((a+b B\bmK)\psi)\sp\ast B\bmK(a+b B\bmK)\psi
\]
\[
=\bar a\psi\sp\ast(\bar a B\bmK+\bar b)\psi
+
\bar b\,\overline{\psi\sp\ast (a +b B\bmK)\psi}
=\bar a^2\psi\sp\ast B\bmK\psi
+2\bar a\bar b\psi\sp\ast\psi
+
\bar b^2\overline{\psi\sp\ast B\bmK\psi}.
\]
The integration of the above charge densities
leads to \eqref{new-lambda}.
The relation \eqref{q-lambda-invariant}
is verified by the explicit computation.
%% \ac{In the remaining part of the proof, $\Lambda$ likely
%% needs to get the complex conjugation}
%% Finally,
%% writing $Q=Q(\psi)$ and $\Lambda=\Lambda(\psi)$,
%% we verify \eqref{q-lambda-invariant}:
%% \[
%% ((\abs{a}^2+\abs{b}^2)Q(\psi)
%% +2\Re\{a \bar b \Lambda(\psi)\})^2
%% -
%% \abs{a^2\Lambda(\psi)
%% +2 a b Q(\psi)+b^2\overline{\Lambda(\psi)}}^2
%% \]
%% \[
%% =
%% (\abs{a}^2+\abs{b}^2)^2 Q^2
%% +4(\abs{a}^2+\abs{b}^2) Q\Re\{a \bar b \Lambda\}
%% +4(\Re\{a \bar b \Lambda\})^2
%% -
%% (a^2\Lambda+2 a b Q+b^2\overline{\Lambda})
%% (\bar a^2\bar \Lambda+2 \bar a \bar b Q+\bar b^2\Lambda)
%% \]
%% \[
%% =
%% (\abs{a}^2+\abs{b}^2)^2 Q^2
%% +4(\abs{a}^2+\abs{b}^2) Q\Re\{a \bar b \Lambda\}
%% +(a \bar b \Lambda+\bar a b \bar\Lambda)^2
%% \]
%% \[
%% -
%% (\abs{a}^4\abs{\Lambda}^2+4 \abs{a}^2\abs{b}^2Q^2
%% +\abs{b}^4\abs{\Lambda}^2)
%% -
%% (a^2\bar b^2\Lambda^2+\bar a^2 b^2\bar \Lambda^2)
%% -2Q(a^2\bar a\bar b\Lambda+\bar a^2 a b\bar\Lambda
%% +b^2\bar b\bar a\bar\Lambda+\bar b^2 b a\Lambda)
%% \]
%% \[
%% =
%% (\abs{a}^2-\abs{b}^2)^2 Q^2
%% +4(\abs{a}^2+\abs{b}^2) Q\Re\{a \bar b \Lambda\}
%% +a^2 \bar b^2 \Lambda^2
%% +2\abs{a}^2\abs{b}^2\abs{\Lambda}^2
%% +\bar a^2 b^2 \bar\Lambda^2
%% \]
%% \[
%% -(\abs{a}^4+\abs{b}^4)\abs{\Lambda}^2
%% -(a^2\bar b^2\Lambda^2+\bar a^2 b^2\bar \Lambda^2)
%% -4Q(\abs{a}^2\Re\{a\bar b\Lambda\}
%% +\abs{b}^2\Re\{a\bar b\Lambda\})
%% \]
%% \[
%% =
%% Q^2+2\abs{a}^2\abs{b}^2\abs{\Lambda}^2
%% -
%% (\abs{a}^4\abs{\Lambda}^2+\abs{b}^4\abs{\Lambda}^2)
%% =
%% Q^2-\abs{\Lambda}^2.
%% \qedhere
%% \]
\end{proof}

\section{Bi-frequency solitary waves}
\label{sect-bi-frequency}

We recall that
$\sigma_r=\frac{x\cdot\sigma}{r}$ (for $r=\abs{x}>0$);
the relations \eqref{general-sigma} imply that
$\sigma_r\sigma_r^\ast=I_{N/2}$.

\begin{lemma}\label{lemma-bi}
Let $n\in\N$, $\omega\in[-m,m]$.
If $v(r)$, $u(r)$ are real-valued functions
which solve \eqref{def-v-u}
with $f$
given by \eqref{f-v-u}
for the nonlinear Dirac equation
(or \eqref{f-not-v-u} for Dirac--Klein--Gordon system),
so that
for any $\xi\in\C^{N/2}$, $\abs{\xi}=1$,
the function
\begin{eqnarray}\label{one-frequency}
\psi(t,x)
=\phi\sb\xi(x)
e^{-\jj\omega t},
\end{eqnarray}
with
\begin{eqnarray}\label{def-phi}
\phi\sb\xi(x)=
\begin{bmatrix}v(r)\xi
\\\jj u(r)\sigma_r \xi
\end{bmatrix},
\qquad
r=\abs{x},
\end{eqnarray}
is a solitary wave solution to the nonlinear Dirac equation
\eqref{nld} (or the Dirac--Klein--Gordon system \eqref{dkg}),
then for any $\Xi,\,\Eta\in\C^{N/2}\setminus\{0\}$,
$\abs{\Xi}^2-\abs{\Eta}^2=1$,
the function
\begin{eqnarray}
\label{bi-frequency-solitary-wave}
\theta\sb{\Xi,\Eta}(t,x)
=
\abs{\Xi}
%%\begin{bmatrix}v(r)\xi\\\jj u(r)\sigma_r\xi\end{bmatrix}
\phi\sb\xi(x)
e^{-\jj\omega t}
+
\abs{\Eta}
%%\begin{bmatrix}-\jj u(r)\sigma_r\sp\ast\eta\\v(r)\eta\end{bmatrix}
\chi\sb\eta(x)
e^{\jj\omega t},
\qquad
\xi=\frac{\Xi}{\abs{\Xi}},
\qquad
\eta=\frac{\Eta}{\abs{\Eta}},
\end{eqnarray}
with
\begin{eqnarray}\label{def-chi}
\chi\sb\eta(x)
=
\begin{bmatrix}
-\jj u(r)\sigma_r\sp\ast\eta
\\
v(r)\eta
\end{bmatrix},
\qquad
r=\abs{x},
\end{eqnarray}
is a solution to \eqref{nld} (or \eqref{dkg}, respectively).
\end{lemma}

\begin{proof}
The lemma is verified by the direct substitution.
First one checks that
\[
\theta\sb{\Xi,\Eta}\sp\ast\beta\theta\sb{\Xi,\Eta}
=
\abs{\Xi}^2(v^2-u^2)
+\abs{\Eta}^2(u^2-v^2)
=v^2-u^2=\phi\sb\xi\sp\ast\beta\phi\sb\xi.
\]
%% with
%% $\phi\sb\xi(x)=\begin{bmatrix}
%% v(r)\xi
%% \\
%% \jj u(r)\sigma_r\xi
%% \end{bmatrix}
%% $.
It remains to prove that the relation
$\omega\phi\sb\xi=D_m\phi\sb\xi-\beta f\phi\sb\xi$
implies the relation
$-\omega\chi\sb\eta=D_m\chi\sb\eta-\beta f\chi\sb\eta$,
with $\chi\sb\eta(x)$ defined in \eqref{def-chi},
%% \begin{eqnarray}\label{def-chi}
%% \chi\sb\eta(x)
%% =
%% \begin{bmatrix}
%% -\jj u(r)\sigma_r\sp\ast\eta
%% \\
%% v(r)\eta
%% \end{bmatrix}
%% ,
%% \qquad
%% r=\abs{x},
%% \end{eqnarray}
for any $\eta\in\C^{N/2}$,
$\abs{\eta}=1$.
With $D_m$ built with the Dirac matrices from \eqref{general-alpha},
the above relations on $\phi\sb\xi$ and $\chi\sb\eta$
are written explicitly as follows:
\begin{eqnarray}
\omega
\begin{bmatrix}
v\xi\\\jj\sigma_r u\xi
\end{bmatrix}
&=&
\begin{bmatrix}
0&-\jj\sigma\sp\ast\cdot\nabla
\\
-\jj\sigma\cdot\nabla&0
\end{bmatrix}
\begin{bmatrix}
v\xi\\\jj\sigma_r u\xi
\end{bmatrix}
+
(m-f)
\begin{bmatrix}
v\xi\\-\jj\sigma_r u\xi
\end{bmatrix},
\nonumber
\\
\nonumber
-\omega
\begin{bmatrix}
-\jj\sigma_r\sp\ast u\eta\\v\eta
\end{bmatrix}
&=&
\begin{bmatrix}
0&-\jj\sigma\sp\ast\cdot\nabla
\\
-\jj\sigma\cdot\nabla&0
\end{bmatrix}
\begin{bmatrix}
-\jj\sigma_r\sp\ast u\eta\\v\eta
\end{bmatrix}
+
(m-f)
\begin{bmatrix}
-\jj\sigma_r\sp\ast u\eta\\-v\eta
\end{bmatrix};
\nonumber
\end{eqnarray}
each of these relations is equivalent to the system \eqref{def-v-u}
(with $f$
given by \eqref{f-v-u}
for the nonlinear Dirac equation
or \eqref{f-not-v-u} for Dirac--Klein--Gordon system)
once we take into account that
%% taking into account that
%% $\sigma\cdot\nabla U=\sigma_r \p_r U
%% =\sigma_r u$
%% and
$\sigma\cdot\nabla v=\sigma_r \p_r v$,
\ $\sigma\sp\ast\cdot\nabla v=\sigma_r\sp\ast \p_r v$,
\ and
\[
(\sigma\sp\ast\cdot\nabla)\sigma_r u
=
(\sigma\sp\ast\cdot\nabla)\sigma_r
\p_r U
=
(\sigma\sp\ast\cdot\nabla)
(\sigma\cdot\nabla)U
=\Delta U=\p_r u+\frac{n-1}{r}u,
\]
where we introduced $U(r)=\int_0^r u(s)\,ds$
and used the identity
$
%%(\sigma\cdot\nabla)(\sigma\sp\ast\cdot\nabla)=
(\sigma\sp\ast\cdot\nabla)(\sigma\cdot\nabla)
=1\sb{N/2}\Delta
$
which follows from \eqref{general-sigma};
similarly,
one has
$(\sigma\cdot\nabla)\sigma_r\sp\ast u
=\p_r u+\frac{n-1}{r}u$.
\end{proof}

Assume that in \eqref{nld}
$f\in C^1(\R\setminus\{0\})\cap C(\R)$,
so that the linearization at a solitary wave makes sense.

\begin{corollary}\label{corollary-bi}
The linearization at a
(one-frequency)
solitary wave
has eigenvalues $\pm 2\omega\jj$
of geometric multiplicity (at least) $N/2$.
\end{corollary}

\begin{proof}
Since
$
\psi(t,x)
=
\big(
(1+\epsilon^2)^{1/2}\phi\sb\xi(x)
+\epsilon \chi\sb\eta(x)e^{2\jj\omega t}
\big)e^{-\jj\omega t},
$
for any
$0<\epsilon\ll 1$
and $\xi,\,\eta\in\C^{N/2}$,
satisfies the nonlinear Dirac equation \eqref{nld},
one concludes that
$r(t,x)=\chi\sb\eta(x)e^{2\jj\omega t}$
is a solution to the nonlinear Dirac equation
%\eqref{nld}
linearized
at $\phi\sb\xi e^{-\jj\omega t}$.
This shows that
$2\omega\jj$
is an eigenvalue of the linearization.
Due to the symmetry of the spectrum
with respect to $\Re\lambda=0$ and $\Im\lambda=0$,
so is $-2\omega\jj$.
\end{proof}

\begin{remark}
The presence of the eigenvalues $\pm 2\omega\jj$
in the spectrum of the linearization of the Soler model
at a solitary wave
was noticed in \cite{MR2892774}
(initially in the one-dimensional case)
and eventually led to the conclusion
that there exist bi-frequency solitary waves.
We note that the existence of such bi-frequency solutions
could have already been deduced
applying the Bogoliubov transformation \eqref{def-g-b}
from \cite{MR0503135}
to one-frequency solitary waves
$\phi(x)e^{-\jj\omega t}$
which were constructed in \cite{PhysRevD.1.2766}.
\end{remark}

We define the solitary manifold
of one- and bi-frequency solutions of the form
\eqref{bi-frequency-solitary-wave}
corresponding to some value $\omega$
by
\begin{eqnarray}
\mathscr{M}\sb\omega
=\left\{
\theta\sb{\Xi,\Eta}(t,x)
\sothat
\quad
\quad\Xi,\,\Eta\in\C^{N/2},
\quad
\abs{\Xi}^2-\abs{\Eta}^2=1
\right\}.
\end{eqnarray}
In general,
the solitary manifold $\mathscr{M}\sb\omega$
can be larger than the orbit of $\phi e^{-\jj\omega t}$
under the action of the available symmetry groups:
$G_{\mathrm{Bogoliubov}}$
defined in \eqref{def-g-b}
and $\mathbf{SO}(n)$;
we denote this orbit by
\[
\mathcal{O}\sb{\phi}
=
\left\{
r g 
\left(\phi e^{-\jj\omega t}\right)
\sothat
r\in \mathbf{SO}(n),
\ \ g\in
G\sb{\mathrm{Bogoliubov}}
\right\}
\subset\mathscr{M}\sb\omega.
\]

\begin{remark}
We do not consider translations and Lorentz boosts,
thus preserving the spatial location of the solitary wave,
just like it is preserved in the definition of $\mathscr{M}\sb\omega$.
\end{remark}

%% \begin{remark}
%% The Bogoliubov group
%% could be generated by several generators
%% $B_i\bmK$,
%% $1\le i\le k$,
%% with $B_i\in\End(\C^N)$.
%% \end{remark}

In lower spatial dimensions $n\le 2$,
when $N=2$,
the orbit $\mathcal{O}\sb{\phi}$ is given by
\[
\mathcal{O}\sb{\phi}
=\{g(\phi(x)e^{-\jj\omega t})
\sothat g\in G\sb{\mathrm{Bogoliubov}}\},
\]
with
$
G\sb{\mathrm{Bogoliubov}}
=\big\{
a+b\sigma_1\bmK,\ a,\,b\in\C,\ \abs{a}^2-\abs{b}^2=1
\big\};
$
thus,
$\mathcal{O}\sb{\phi}$
coincides with the solitary manifold
\[
\mathscr{M}\sb\omega
=\left\{
a\begin{bmatrix}v(x)\\u(x)\end{bmatrix}e^{-\jj\omega t}
+
b\begin{bmatrix}u(x)\\v(x)\end{bmatrix}e^{\jj\omega t}
\sothat
a,\,b\in\C,\ \abs{a}^2-\abs{b}^2=1
\right\}.
\]
In three spatial dimensions,
$n=3$ and $N=4$,
the solitary manifold $\mathscr{M}\sb\omega$
is larger than
the orbit $\mathcal{O}\sb{\phi}$
of $\phi(x)e^{-\jj\omega t}$
under the action of the available symmetry groups:
spatial rotations
$\mathbf{SO}(n)$
and the Bogoliubov group
$G_{\mathrm{Bogoliubov}}$
given by elements of the form
$a+b(-\jj\gamma^2)\bmK$,
where $a,\,b\in\C$, $\abs{a}^2-\abs{b}^2=1$.
Indeed, one has
$\mathcal{O}\sb{\phi}\subsetneq\mathscr{M}\sb\omega$,
since
\[
\dim\sb\R \mathcal{O}\sb{\phi}
=5
<\dim\sb\R \mathbf{SO}(3)+\dim\sb\R G_{\mathrm{Bogoliubov}}
=3+3,
\]
\[
\dim\sb\R \mathscr{M}\sb\omega
=
\dim\sb\R\{
(\Xi,\Eta)\in\C^{N/2}\times\C^{N/2}
\sothat
\abs{\Xi}^2-\abs{\Eta}^2=1
\}
=
2N-1=7.
\]
Note that
in the above inequality
for $\dim\sb\R \mathcal{O}\sb{\phi}$
one has ``strictly smaller'',
since the generator corresponding to
the standard $\mathbf{U}(1)$-invariance
which enters the Lie algebra of $G_{\mathrm{Bogoliubov}}$
also coincides with the generator of rotation around $z$-axis.

%%A similar situation takes place in four spatial dimensions:
In the case $n=4$, $N=4$,
the symmetry group simplifies to $\mathbf{U}(1)$
since the generator
$-\jj\gamma^2\bmK$ is no longer available:
the Dirac operator now contains
$-\jj\alpha^4\p_{x\sp 4}=\beta\gamma^5\p_{x\sp 4}$
(with
$\alpha^4:=\begin{bmatrix}0&-\jj 1\sb{2}\\\jj 1\sb{2}&0\end{bmatrix}
=-\jj\beta\gamma^5$),
which breaks the anticommutation of $D_m$ with
$-\jj\gamma^2\bmK$:
\[
\{
-\jj\alpha^4,
\jj\alpha^2\beta\bmK
\}
=
\{
\beta\gamma^5,
\jj\alpha^2\beta\bmK
\}
=
\beta\gamma^5
\jj\alpha^2\beta\bmK
+
\jj\alpha^2\beta\bmK
\beta\gamma^5
\]
\[
=
\jj\gamma^5\alpha^2\bmK
+\jj\alpha^2\gamma^5\bmK
=
2\jj\alpha^2\gamma^5\bmK
\ne 0.
\]
Again,
$\mathcal{O}\sb{\phi}\subsetneq\mathscr{M}\sb\omega$
since
\[
\dim\sb\R \mathcal{O}\sb{\phi}
\le
\dim\sb\R \mathbf{SO}(4)
+\dim\sb\R G_{\mathrm{Bogoliubov}}
\le 6+1,
\qquad
\dim\sb\R\mathscr{M}\sb\omega
=
2N-1=7,
\]
with the Lie algebras of
$\mathbf{SO}(4)$
and $G_{\mathrm{Bogoliubov}}$
sharing one element (a generator of the
standard $\mathbf{U}(1)$-symmetry).
Moreover, the action of $\mathbf{SO}(4)$
($\dim_{\R}=6$)
on
$\C^2$
($\dim_{\R}=4$)
could not be \emph{faithful};
the orbit of an element $\xi\in\C^2$
under the action of $\mathbf{SO}(4)$
is only three-dimensional.
As a result, in the case $n=4$, $N=4$, one has
\[
\dim\sb\R\mathcal{O}\sb{\phi}=3,
\qquad
\dim\sb\R\mathscr{M}\sb\omega=2N-1=7.
\]

\begin{remark}
We can rephrase the above situation in the following way.
When moving from $n=3$ to $n=4$,
additional rotations in $\R^4$ do not add to the
orbit of $\xi\in\C^2$ which has already been of maximal
dimension
when $n=3$
(which equals three:
it is the real dimension of the unit sphere in $\C^2$),
while the loss
of the generator $B\bmK$
from the Bogoliubov group
led to the loss of two real dimensions
of the orbit $\mathcal{O}\sb{\phi}$.
\end{remark}

\begin{remark}
Let us briefly discuss the pseudo-scalar theories
in spatial dimension $n=3$.
Instead of the (scalar) Yukawa interaction,
given by the term
$\phi\bar\psi\psi$ in the Lagrangian,
one can consider pseudoscalar interaction,
introducing the term
$\phi\bar\psi\jj\gamma^5\psi$,
which we write as
$-\phi\psi\sp\ast\alpha^4\psi$
with
\[
\alpha^4=-\jj\beta\gamma^5
=
\begin{bmatrix}0&-\jj I_2\\\jj I_2&0\end{bmatrix}.
\]
The Bogoliubov symmetry $\mathbf{SU}(1,1)$
is no longer present in a model with such an interaction.
For $g=a-\jj b\gamma^2\bmK$, $\abs{a}^2-\abs{b}^2=1$,
omitting taking the real part in the intermediate computations,
one has:
\begin{eqnarray}
\nonumber
&&
(g\psi)\sp\ast\alpha^4(g\psi)
=\Re (g\psi)\sp\ast\alpha^4(g\psi)
=\Re ((a-\jj b\gamma^2\bmK)\psi)\sp\ast\alpha^4(a-\jj b\gamma^2\bmK)\psi
\\
\nonumber
&&
=
\Re
\psi\sp\ast
(\bar a-\bmK \jj \bar b\gamma^2)\alpha^4
(a-\jj b\gamma^2\bmK)\psi
=
\Re
\psi\sp\ast
(\bar a-\jj b\gamma^2\bmK)
\alpha^4
(a-\jj b\gamma^2\bmK)\psi
\\
\nonumber
&&
=
\Re
\psi\sp\ast
\alpha^4
(\bar a-\jj b\gamma^2\bmK)
(a-\jj b\gamma^2\bmK)\psi
\\
\nonumber
&&
=
\Re
\psi\sp\ast
\alpha^4
(\abs{a}^2+\abs{b}^2
-2\jj\bar a b\gamma^2\bmK)
\psi
=
\psi\sp\ast
\alpha^4
(\abs{a}^2+\abs{b}^2)
\psi,
\end{eqnarray}
which in general is different from
$\psi\sp\ast\alpha^4\psi$.
Above, in the last equality,
we took into account that
the matrix
$\alpha^4\gamma^2
=\begin{bmatrix}\jj\sigma_2&0\\0&\jj\sigma_2\end{bmatrix}
$
is antisymmetric,
hence
$
\psi\sp\ast\alpha^4\gamma^2\bmK\psi
=
(\bmK\psi)\sp{\mathrm{t}}\alpha^4\gamma^2\bmK\psi
=0
$.
%% \ac{What about the solitary waves of the form
%% \eqref{bi-frequency-solitary-wave}??}
\end{remark}

\section{Spectral stability of bi-frequency solitary waves}
\label{sect-bi-frequency-stability}

\begin{definition}\label{def-linear-stability}
We say that the bi-frequency solitary wave solution
$\theta\sb{\Xi,\Eta}(t,x)$
(see \eqref{bi-frequency-solitary-wave})
to
\eqref{nld} or \eqref{dkg}
is \emph{linearly unstable}
if there are nonzero functions $\rho_j\in L^2(\R^n,\C^N)$
and numbers $\Lambda_j\in\C$,
$1\le j\le J$, $J\ge 1$,
with
$\Lambda_i\ne\Lambda_j$ except when $i=j$
and with $\Re\Lambda_j>0$,
such that
\begin{eqnarray}\label{ansatz-unstable}
\theta\sb{\Xi,\Eta}(t,x)
+\epsilon
\sum\sb{j=1}\sp{J}\rho_j(x)e^{\Lambda_j t}
\end{eqnarray}
solves
\eqref{nld}
or \eqref{dkg}
up to $o(\epsilon)$,
$0<\epsilon\ll 1$.
Otherwise,
we call the bi-frequency solitary wave solution
\emph{spectrally stable}.
\end{definition}

\begin{theorem}\label{theorem-equiv}
Let $n\le 4$, $N=2$ or $N=4$.
Let \eqref{u-less-v} be satisfied.
%% \[
%% a\phi\sb\xi(x)e^{-\jj\omega t}
%% +
%% b\chi\sb\xi(x)e^{\jj\omega t},
%% \]
%% with $\xi,\,\eta\in\C^{N/2}$,
%% $\abs{\xi}=\abs{\eta}=1$,
%% $a,\,b\in\C$,
%% $\abs{a}^2-\abs{b}^2=1$,
Then the bi-frequency solitary wave
\eqref{bi-frequency-solitary-wave}
is spectrally stable as long as
the corresponding one-frequency solitary wave solution
\eqref{one-frequency}
%% $\phi\sb\xi e^{-\jj\omega t}$
is spectrally stable.
\end{theorem}

\begin{proof}
Let us first give the proof of Theorem~\ref{theorem-equiv}
in the simple case,
\[
n\le 2,\qquad N=2,
\]
when the argument could be based
on reduction
with the aid of the Bogoliubov transformation.
In this case,
by the above considerations,
the solitary manifold $\mathscr{M}\sb\omega$
coincides with the orbit
$\mathcal{O}\sb{\phi}$
of a one-frequency solitary wave
under the action of the symmetry group of the equation.
Therefore,
a bi-frequency wave
$\theta(t,x)=a\phi(x)e^{-\jj\omega t}+b\chi(x)e^{\jj\omega t}$,
$a,\,b\in\C$, $\abs{a}^2-\abs{b}^2=1$,
can be written in the form
$\theta(t,x)
=a\phi e^{-\jj\omega t}
+b B(\bmK\phi)e^{\jj\omega t}
=g\left[\phi(x)e^{-\jj\omega t}\right]
$.
The perturbation \eqref{ansatz-unstable}
of the bi-frequency solitary wave
can be written in the form
\[
a\phi e^{-\jj\omega t}
+
b\chi e^{-\jj\omega t}
+\sum\sb{j=1}\sp{J}\rho_j(x)e^{\Lambda_j t}
=
g\left[
(\phi(x)+\varrho(t,x))e^{-\jj\omega t}\right],
\]
where we take
$\varrho(t,x)
=e^{\jj\omega t}
g^{-1}\left[\sum\sb{j=1}\sp{J}\rho_j(x)e^{\Lambda_j t}\right]$,
with
$g^{-1}=\bar a-b B\bmK$.
The above relation shows that
the exponential growth of
$\sum\sb{j=1}\sp{J}\rho_j(x)e^{\Lambda_j t}$
is in one-to-one correspondence to the exponential growth of
$\varrho(t,x)$.
As a result,
the spectral stability of the bi-frequency wave
$a\phi(x)e^{-\jj\omega t}+b\chi(x)e^{\jj\omega t}$
(cf. Definition~\ref{def-linear-stability})
takes place if and only if
the corresponding one-frequency solitary wave
$\phi(x)e^{-\jj\omega t}$
is spectrally stable.
This completes the proof in the case $n\le 2$, $N=2$.

\medskip

Now we assume that
\[
n\le 4,
\qquad
N=4.
\]
Given
$\Xi,\,\Eta\in\C^{N/2}\setminus\{0\}$ such that
$\abs{\Xi}^2-\abs{\Eta}^2=1$,
let
\begin{eqnarray}\label{bsw}
\theta\sb{\Xi,\Eta}(t,x)
=
a\phi\sb{\xi}(x)e^{-\jj\omega t}
+b\chi\sb{\eta}(x)e^{\jj\omega t},
\end{eqnarray}
with
$a=\abs{\Xi}$,
$b=\abs{\Eta}$,
$\xi=\Xi/\abs{\Xi}$
and $\eta=\Eta/\abs{\Eta}$
(with $\phi\sb\xi$ and $\chi\sb\eta$
from \eqref{def-phi} and \eqref{def-chi},
respectively)
be a bi-frequency solitary wave.
We will consider the perturbation of this solitary wave
in the form
\[
\psi(t,x)
=
a(\phi\sb{\xi}(x)+\rho(t,x))e^{-\jj\omega t}
+b(\chi\sb{\eta}(x)+\sigma(t,x))e^{\jj\omega t},
\]
where we will impose the following condition
on $\rho(t,x)$ and $\sigma(t,x)$:
\begin{eqnarray}\label{no-frequency-mixing-0}
\bar a b \phi_\xi\sp\ast\beta \sigma
+\bar a  b \rho\sp\ast\beta\chi_\eta=0.
\end{eqnarray}
If \eqref{no-frequency-mixing-0} is satisfied,
we will say that there is no \emph{frequency mixing};
in this case,
$\psi\sp\ast\beta\psi$ does not
contribute terms with factors
$e^{-2\jj\omega t}$ or $e^{2\jj\omega t}$.
We will show that indeed there is a way to split the
perturbation into $\rho(t,x)$ and $\sigma(t,x)$ so that
\eqref{no-frequency-mixing-0} is satisfied
(see
Proposition~\ref{prop-invariant-subspace}
and Remark~\ref{remark-no-mixing} below).

Let
$(\e_j)\sb{1\le j\le N/2}$ be the standard basis
in $\C^{N/2}$ and let
$R,\,S\in\mathbf{SU}(N/2)$
be such that
$\xi=R\e_1$, $\eta=S\e_1$.
Denote
\begin{eqnarray}\label{def-phi-chi-j}
\phi_j(x)
=\begin{bmatrix}v(r) R\e_j\\\jj u(r)\sigma_r R\e_j\end{bmatrix},
\qquad
\chi_j(x)
=\begin{bmatrix}-\jj u(r)\sigma_r\sp\ast S\e_j\\v(r) S\e_j\end{bmatrix},
\qquad
1\le j\le N/2.
\end{eqnarray}
(Above, we do not indicate the dependence of $v$, $u$ of $\omega$.)
We consider the perturbation of the bi-frequency solitary wave \eqref{bsw}
in the form
\begin{eqnarray}\label{linearization-ansatz}
\psi(t,x)
=
a\Big(
\phi_1(x)+\sum\sb{j=1}\sp{N/2}
\big(p_j(t,x)
\phi_j(x)
+q_j(t,x)
\chi_j(x)
\big)
\Big)e^{-\jj\omega t}
\nonumber
\\
\quad +\,b\Big(\chi_1(x)
+\sum\sb{j=1}\sp{N/2}
\big(r_j(t,x)
\phi_j(x)
+s_j(t,x)
\chi_j(x)\big)\Big)e^{\jj\omega t}.
\end{eqnarray}
Above, $p_j,\,q_j,\,r_j,\,s_j$,
$1\le j\le N/2$, are complex scalar-valued
functions of $x$ and $t$.
The condition \eqref{no-frequency-mixing-0}
of the absence of \emph{frequency mixing}
takes the form
\begin{eqnarray}\label{no-frequency-mixing}
\bar a b
\big(
\,\overline{q_1(t,x)}-r_1(t,x)
\big)=0.
\end{eqnarray}
As long as \eqref{no-frequency-mixing}
is satisfied,
the linearized terms in the expansion of
$\psi\sp\ast\beta\psi$
do not contain factors
$e^{\pm 2\jj\omega t}$;
the ones that are left are given by
\begin{eqnarray}\label{two-re}
2\Re(\dots)
=2\Re\big(
\abs{a}^2\phi_1\sp\ast\beta\phi_1 \bar p_1
+\abs{b}^2
\chi_1\sp\ast\beta\chi_1 s_1
\big)
=2\phi_1\sp\ast\beta\phi_1
\Re(\abs{a}^2 \bar p_1-\abs{b}^2 s_1).
\end{eqnarray}
The linearized equation will contain two groups of terms,
with factors $e^{\pm \jj\omega t}$;
to satisfy the linearized equation,
it is enough to equate these groups
separately.
The terms with the factor $e^{-\jj\omega t}$:
\begin{eqnarray}
\label{rel-1}
\jj\p_t p_k \phi_k
+\jj\p_t q_k \chi_k
-2\omega q_k \chi_k
=
(D_0 p_k)\phi_k
+(D_0 q_k)\chi_k
-2f'\Re(\dots)\beta\phi_1;
\end{eqnarray}
the terms with the factor $e^{\jj\omega t}$:
\begin{eqnarray}
\label{rel-2}
\jj\p_t r_k \phi_k
+\jj\p_t s_k \chi_k
+2\omega r_k \phi_k
=
(D_0 r_k)\phi_k
+(D_0 s_k)\chi_k
-2f'\Re(\dots)\beta\chi_1.
\end{eqnarray}
%% Acting on \eqref{rel-1}
%% with $B\bmK$ and taking into account
%% that
%% $B\bmK\phi_j=\chi_j$,
%% $B\bmK\chi_j=\phi_j$,
%% one has:
%% \begin{eqnarray}
%% \label{rel-1bc}
%% -\jj\p_t{\bar p_k} \chi_k
%% -\jj\p_t \bar q_k \phi_k
%% -2\omega q_k \phi_k
%% =
%% -(D_0 \bar p_k)\chi_k
%% -(D_0 \bar q_k)\phi_k
%% +2f'\Re(\dots)\beta\chi_1,
%% \end{eqnarray}
%% which agrees with \eqref{rel-2}
%% as long as .....

\begin{remark}
When deriving the above equations,
we eliminated terms with the derivatives of
$\phi_j$ and $\chi\sb j$
by using the stationary Dirac equations
satisfied by $\phi_j$ and $\chi_j$:
\[
\omega\phi_j=(D_m-\beta f)\phi_j,
\qquad
-\omega\chi_j=(D_m-\beta f)\chi_j,
\qquad
1\le j\le N/2,
\]
where $f=f(\tau)$
is evaluated at
$\tau=v^2-u^2=\phi_1\sp\ast\beta\phi\sb 1$.
\end{remark}

\begin{proposition}\label{prop-invariant-subspace}
The system
\eqref{rel-1},
\eqref{rel-2}
is invariant in the subspace
specified by the relations
\begin{eqnarray}\label{invariant-subspace}
r_j=\bar q_j,
\qquad
s_j=\bar p_j,
\qquad
1\le j\le N/2.
\end{eqnarray}
\end{proposition}

\begin{proof}
We claim that if
$p_j$, $q_j$, $r_j$, and $s_j$, $1\le j\le 2$,
satisfy \eqref{invariant-subspace},
then equations
\eqref{rel-1} and \eqref{rel-2}
yield
\begin{eqnarray}\label{p-t-u-p-t-v}
\p\sb t r_j=\p\sb t\bar q_j,
\qquad
\p\sb t s_j=\p\sb t\bar p_j,
\qquad
1\le j\le N/2.
\end{eqnarray}

\noindent
Multiplying the relation
\eqref{rel-1}
by $\beta$ and coupling with $\phi_j$
(in the $\C^N$-sense; no integration in $x$):
\begin{eqnarray}
\label{a31a}
\jj\p_t p_k
\phi_j\sp\ast\beta\phi\sb k
%+\jj \p_t q_k\phi_j\sp\ast\beta\chi\sb k
%-2\omega q_k \phi_j\sp\ast\beta\chi\sb k
=
-\jj\phi_j\beta\alpha^i\phi\sb k \p_i p_k
-\jj\phi_j\beta\alpha^i\chi\sb k \p_i q_k
-2f'\Re(\dots)\phi_j\sp\ast\phi_1.
\end{eqnarray}
We took into account
that
$\phi\sb j\sp\ast\beta\chi\sb k=0$
for all
$1\le j,\,k\le N/2$
(cf. Lemma~\ref{lemma-matrix-elements-general} below).
Multiplying the relation \eqref{rel-2}
by $\beta$ and coupling with $\chi_j$:
\begin{eqnarray}
\label{a31b}
%\jj\p_t r_k \chi_j\sp\ast\beta\phi\sb k
%+
\jj\p_t s_k \chi_j\sp\ast\beta\chi\sb k
%+2\omega r_k \chi_j\sp\ast\beta\phi_k
=
-\jj\chi_j\sp\ast\beta\alpha^i\phi_k \p_i r_k
-\jj\chi_j\sp\ast\beta\alpha^i\chi_k \p_i s_k
-2f'\Re(\dots)\chi_j\sp\ast \chi_1.
\end{eqnarray}
%% It is enough to show that
%% the former relation is the complex conjugate of the latter.
Multiplying the relation
\eqref{rel-1}
by $\beta$ and coupling with $\chi_j$
one has:
\begin{eqnarray}
\label{b31a}
%\jj\p_t p_k
%\chi_j\sp\ast\beta\phi\sb k
%+
&&
\jj \p_t q_k\chi_j\sp\ast\beta\chi\sb k
-2\omega q_k \chi_j\sp\ast\beta\chi\sb k
\nonumber
\\
&&
=
-\jj\chi_j\beta\alpha^i\phi\sb k \p_i p_k
-\jj\chi_j\beta\alpha^i\chi\sb k \p_i q_k
-2f'\Re(\dots)\chi_j\sp\ast\phi_1.
\end{eqnarray}
Multiplying the relation \eqref{rel-2}
by $\beta$ and coupling with $\phi_j$:
\begin{eqnarray}
\label{b31b}
&&
\jj\p_t r_k \phi_j\sp\ast\beta\phi\sb k
%+\jj\p_t s_k \phi_j\sp\ast\beta\chi\sb k
+2\omega r_k \phi_j\sp\ast\beta\phi_k
\nonumber
\\
&&
=
-\jj\phi_j\sp\ast\beta\alpha^i\phi_k \p_i r_k
-\jj\phi_j\sp\ast\beta\alpha^i\chi_k \p_i s_k
-2f'\Re(\dots)\phi_j\sp\ast \chi_1.
\end{eqnarray}
The proof of Proposition~\ref{prop-invariant-subspace}
will follow
if we prove that
\eqref{a31a} and \eqref{a31b}
are complex conjugates of each other,
and that so are
\eqref{b31a} and \eqref{b31b}.

\begin{lemma}\label{lemma-matrix-elements-general}
For $1\le j,k\le N/2$,
\begin{eqnarray}\label{matrix-elements-one}
\phi_j\sp\ast\phi_k=\overline{\chi_j\sp\ast\chi_k},
\qquad
\phi_j\sp\ast\chi_k=\overline{\chi_j\sp\ast\phi_k};
\end{eqnarray}
\begin{eqnarray}\label{matrix-elements-beta}
\phi_j\sp\ast\beta\phi_k
=-\overline{\chi_j\sp\ast\beta\chi_k},
\qquad
\phi_j\sp\ast\beta\chi_k
=0,
\qquad
\chi_j\sp\ast\beta\phi_k=0;
\end{eqnarray}
\begin{eqnarray}\label{matrix-elements-beta-alpha}
\phi_j\sp\ast(-\jj\beta\alpha^i)\phi_k
=\overline{\chi_j\sp\ast(-\jj\beta\alpha^i)\chi_k},
\quad
\phi_j\sp\ast(-\jj\beta\alpha^i)\chi_k
=\overline{\chi_j\sp\ast(-\jj\beta\alpha^i)\phi_k},
\quad
1\le i\le n.
\end{eqnarray}
\end{lemma}

\begin{proof}
We note that
$\sigma_r\sp\ast\sigma_r=\sigma_r\sigma_r\sp\ast=1$
(cf. \eqref{def-sigma-r}, \eqref{general-sigma}).
We have:
\[
\phi_j\sp\ast\phi_k
=
\begin{bmatrix}v R\e_j\\\jj u\sigma_r R\e_j\end{bmatrix}\sp\ast
\begin{bmatrix}v R\e_k\\\jj u\sigma_r R\e_k\end{bmatrix}
=v^2\e_j\sp\ast\e_k+u^2\e_j\sp\ast R\sp\ast\sigma_r\sp\ast\sigma_r R\e_k
=(v^2+u^2)\e_j\sp\ast\e_k,
\]
\[
\chi_j\sp\ast\chi_k
=
\begin{bmatrix}-\jj u\sigma_r\sp\ast S\e_j\\v S\e_j\end{bmatrix}\sp\ast
\begin{bmatrix}-\jj u\sigma_r\sp\ast S\e_k\\v S\e_k\end{bmatrix}
=u^2\e_j\sp\ast S\sp\ast\sigma_r\sigma_r\sp\ast S\e_k+v^2\e_j\sp\ast\e_k
=(v^2+u^2)\e_j\sp\ast\e_k,
\]
\begin{eqnarray}\label{phi-chi}
\phi_j\sp\ast\chi_k
=
\begin{bmatrix}v R\e_j\\\jj u\sigma_r R\e_j\end{bmatrix}\sp\ast
\begin{bmatrix}-\jj u\sigma_r\sp\ast S\e_k\\v S\e_k\end{bmatrix}
=-2\jj v u\e_j\sp\ast R\sp\ast \sigma_r\sp\ast S\e_k,
\end{eqnarray}
\begin{eqnarray}\label{chi-phi}
\chi_j\sp\ast\phi_k
=
\begin{bmatrix}-\jj u\sigma_r\sp\ast S\e_j\\v S\e_j\end{bmatrix}\sp\ast
\begin{bmatrix}v R\e_k\\\jj u\sigma_r R\e_k\end{bmatrix}
=2\jj v u\e_j\sp\ast S\sp\ast \sigma_r R\e_k.
\end{eqnarray}
To show that \eqref{phi-chi} and \eqref{chi-phi}
are complex conjugates of each other,
it suffices to mention the identities
\begin{eqnarray}\label{identities}
\overline{\e_1\sp\ast R^*\sigma_i^* S\e_2}
=
\e_1\sp\ast S^*\sigma_i R\e_2,
\qquad
\overline{\e_2\sp\ast R^*\sigma_i^* S\e_1}
=
\e_2\sp\ast S^*\sigma_i R\e_1,
\qquad
1\le i\le n,
\end{eqnarray}
valid for all $R,\,S\in\mathbf{SU}(2)$.
(Indeed,
$M:=R^\ast\sigma_i\sp\ast S$
satisfies
$M\in\mathbf{U}(2)$,
$\det M=-1$;
for such matrices, one has
$M_{12}=\overline{M_{21}}$.)
%% $
%% R\sp\ast\sigma_r\sp\ast S\in
%% \mathbf{SU}(2)=\left\{
%% \begin{bmatrix}a&b\\-\bar b&\bar a\end{bmatrix}:
%% \;a,\,b\in\C,\ \abs{a}^2-\abs{b}^2=1
%% \right\}$.
This proves relations \eqref{matrix-elements-one}.

\begin{remark}
The relation $M_{j i}=\overline{M_{i j}}$ for $i\ne j$
is no longer true for
$M=R\sp\ast\sigma_i\sp\ast S$
with $\sigma_i$ the equivalent of the Pauli matrix
of size $N/2$
and with $R,\,S\in\mathbf{SU}(N/2)$
with $N>4$,
seemingly
limiting the present approach to four-component spinors.
%% For $N>4$,
%% the decomposition of the perturbation
%% into invariant subspaces which avoids frequency mixing
%% is yet to be found.
\end{remark}

We continue:
\begin{eqnarray}
\nonumber
\phi_j\sp\ast\beta\phi_k
&=&
\begin{bmatrix}v R\e_j\\\jj u\sigma_r R\e_j\end{bmatrix}\sp\ast
\begin{bmatrix}1\sb{N/2}&0\\0&-1\sb{N/2}\end{bmatrix}
\begin{bmatrix}v R\e_k\\\jj u\sigma_r R\e_k\end{bmatrix}
\\
\nonumber
&=&v^2\e_j\sp\ast\e_k-u^2\e_j\sp\ast R\sp\ast\sigma_r\sp\ast\sigma_r R\e_k
=(v^2-u^2)\e_j\sp\ast\e_k,
\end{eqnarray}
\begin{eqnarray}
\nonumber
\chi_j\sp\ast\beta\chi_k
&=&
\begin{bmatrix}-\jj u\sigma_r\sp\ast S\e_j\\v S\e_j\end{bmatrix}\sp\ast
\begin{bmatrix}1\sb{N/2}&0\\0&-1\sb{N/2}\end{bmatrix}
\begin{bmatrix}-\jj u\sigma_r\sp\ast S\e_k\\v S\e_k\end{bmatrix}
\\
\nonumber
&=&
u^2\e_j\sp\ast S\sp\ast\sigma_r\sigma_r\sp\ast S\e_k-v^2\e_j\sp\ast\e_k
=(u^2-v^2)\e_j\sp\ast \e_k,
\end{eqnarray}
\[
\phi_j\sp\ast\beta\chi_k
=
\begin{bmatrix}v R\e_j\\\jj u\sigma_r R\e_j\end{bmatrix}\sp\ast
\begin{bmatrix}1\sb{N/2}&0\\0&-1\sb{N/2}\end{bmatrix}
\begin{bmatrix}-\jj u\sigma_r\sp\ast S\e_k\\v S\e_k\end{bmatrix}
=0.
\]
This proves the relations \eqref{matrix-elements-beta}.
Finally, we prove \eqref{matrix-elements-beta-alpha}:
\begin{eqnarray}
\nonumber
\phi_j\sp\ast\beta\alpha^i\phi_k
&=&
\begin{bmatrix}v R\e_j\\\jj u\sigma_r R\e_j\end{bmatrix}\sp\ast
\begin{bmatrix}0&\sigma_i^\ast\\-\sigma_i&0\end{bmatrix}
\begin{bmatrix}v R\e_k\\\jj u\sigma_r R\e_k\end{bmatrix}
\\
\nonumber
&=&
\jj v u \e_j\sp\ast R\sp\ast\sigma_i\sp\ast \sigma_r R\e_k
+
\jj v u \e_j\sp\ast R\sp\ast\sigma_r\sp\ast\sigma_i R\e_k
=
2\jj v u \frac{x^i}{r}\e_j\sp\ast\e_k,
\end{eqnarray}
\begin{eqnarray}
\nonumber
\chi_j\sp\ast\beta\alpha^i\chi_k
&=&
\begin{bmatrix}-\jj u\sigma_r\sp\ast S\e_j\\v S\e_j\end{bmatrix}\sp\ast
\begin{bmatrix}0&\sigma_i^\ast\\-\sigma_i&0\end{bmatrix}
\begin{bmatrix}-\jj u\sigma_r\sp\ast S\e_k\\v S\e_k\end{bmatrix}
\\
\nonumber
&=&
\jj u v \e_j\sp\ast S\sp\ast\sigma_r\sigma_i\sp\ast S\e_k
+
\jj u v \e_j\sp\ast S\sp\ast\sigma_i\sigma_r\sp\ast S\e_k
=
2\jj u v \frac{x^i}{r}\e_j\sp\ast\e_k,
\end{eqnarray}
\begin{eqnarray}
\nonumber
\phi_j\sp\ast\beta\alpha^i\chi_k
&=&
\begin{bmatrix}v R\e_j\\\jj u\sigma_r R\e_j\end{bmatrix}\sp\ast
\begin{bmatrix}0&\sigma_i^\ast\\-\sigma_i&0\end{bmatrix}
\begin{bmatrix}-\jj u\sigma_r\sp\ast S\e_k\\v S\e_k\end{bmatrix}
\\
\nonumber
&=&
v^2\e_j\sp\ast R\sp\ast\sigma_i\sp\ast S\e_k
+
u^2\e_j\sp\ast R\sp\ast
\sigma_r^\ast
\sigma_i
\sigma_r^\ast
S\e_k,
\end{eqnarray}
\begin{eqnarray}
\nonumber
\chi_j\sp\ast\beta\alpha^i\phi_k
&=&
\begin{bmatrix}-\jj u\sigma_r\sp\ast S\e_j\\v S\e_j\end{bmatrix}
\sp\ast
\begin{bmatrix}0&\sigma_i^\ast\\-\sigma_i&0\end{bmatrix}
\begin{bmatrix}v R\e_k\\\jj u\sigma_r R\e_k\end{bmatrix}
\\
&=&
-v^2\e_j\sp\ast S\sp\ast\sigma_i R\e_k
-
u^2\e_j\sp\ast S\sp\ast
\sigma_r
\sigma_i\sp\ast
\sigma_r
R\e_k.
\end{eqnarray}
To argue that the last two lines are anti- complex conjugates,
we note that
\[
\sigma_r\sp\ast\sigma_i\sigma_r\sp\ast
=
\sigma_r\sp\ast\Big(2\frac{x_i}{r}-\sigma_r\sigma_i\sp\ast\Big)
=
2\sigma_r\sp\ast\frac{x_i}{r}-\sigma_i\sp\ast,
\qquad
1\le i\le n,
\]
and then use the same reasoning as above
(when showing that \eqref{phi-chi} and \eqref{chi-phi}
are complex conjugates,
basing on the identities \eqref{identities}).
\end{proof}

Lemma~\ref{lemma-matrix-elements-general}
proves \eqref{p-t-u-p-t-v},
finishing the proof of
Proposition~\ref{prop-invariant-subspace}.
\end{proof}

\begin{remark}\label{remark-no-mixing}
We note that
the relation \eqref{no-frequency-mixing} is satisfied
in the invariant subspace
described in Proposition~\ref{prop-invariant-subspace},
and thus there is no \emph{frequency mixing} in this subspace:
given $\psi$ of the form
\eqref{linearization-ansatz}
with
$r_j=\bar q_j$,
$s_j=\bar p_j$,
$1\le j\le N/2$,
the expression $\psi\sp\ast\beta\psi$
does not contain terms with the factors
$e^{\pm 2\jj\omega t}$.
\end{remark}

We introduce the following functions:
\begin{eqnarray}\label{def-varphi}
\varPhi_j(x)=\frac{1}{v(r)}\phi_j(x),
\qquad
\Chi_j(x)=\frac{1}{v(r)}\chi_j(x),
\qquad
1\le j\le N/2;
\end{eqnarray}
at each $x\in\R^n$,
these functions form a basis in $\C^N$.

\begin{lemma}\label{lemma-lin-independent}
Let \eqref{u-less-v} be satisfied.
Then $\varPhi_j(x)$, $1\le j\le N/2$,
and
$\Chi_j(x)$, $1\le j\le N/2$,
are linearly independent, uniformly in $x$.
\end{lemma}

\begin{proof}
For any $1\le j\le N/2$ and any $x\in\R^n$,
$\norm{\varPhi_j(x)}\sb{\C^N}+\norm{\Chi_j(x)}\sb{\C^N}
\le C<\infty$
since $\abs{u(r)/v(r)}\le c<1$
by \eqref{u-less-v},
while
$\det[\varPhi_1\,\varPhi_2\,\Chi_1\,\Chi_2]$
is given by
\[
\det
\left[
\begin{pmatrix}\e_1\\\jj\frac{u}{v}\sigma_r\e_1\end{pmatrix}
\begin{pmatrix}\e_2\\\jj\frac{u}{v}\sigma_r\e_2\end{pmatrix}
\begin{pmatrix}-\jj\frac{u}{v}\sigma_r\sp\ast\e_1\\\e_1\end{pmatrix}
\begin{pmatrix}-\jj\frac{u}{v}\sigma_r\sp\ast\e_2\\\e_2\end{pmatrix}
\right]
\]
\[
=
\det
\begin{bmatrix}
I_2&-\jj\frac{u}{v}\sigma_r\sp\ast
\\
\jj\frac{u}{v}\sigma_r&I_2
\end{bmatrix}
=
\det
\Big(
I_2-\frac{u^2}{v^2}\sigma_r\sp\ast\sigma_r
\Big);
\]
in the last equality,
one can use the Schur complement
to compute
the determinant
of a matrix written in the block form.
Using \eqref{u-less-v}
and taking into account that
$\sigma_r\sp\ast\sigma_r=I_2$,
one concludes that the right-hand side of the above is separated from zero
uniformly in $x\in\R^n$.
\end{proof}

The perturbation of a bi-frequency solitary wave
could be rewritten as follows (cf. \eqref{linearization-ansatz}):
\begin{eqnarray}\label{linearization-ansatz-1}
\psi(t,x)
=
a\Big(\phi_1(x)+\sum\sb{j=1}\sp{N/2}\big(P_j(t,x)\varPhi_j(x)+Q_j(t,x)\Chi_j(x)\big)\Big)
e^{-\jj\omega t}
\nonumber
\\
+\,
b\Big(\chi_1(x)+\sum\sb{j=1}\sp{N/2}\big(\bar Q_j(t,x)\varPhi_j(x)+\bar P_j(t,x)\Chi_j(x)\big)\Big)
e^{\jj\omega t}.
\end{eqnarray}
Taking into account
\eqref{def-varphi},
we note that $P_j$ and $Q_j$ in the above formula
differ from $p_j$ and $q_j$ in \eqref{linearization-ansatz}
by the factor of $v(r)$:
\[
P_j(t,x)= v(r) p_j(t,x),
\qquad
Q_j(t,x)= v(r) q_j(t,x),
\qquad
1\le j\le N/2.
\]

We claim that at the initial moment
there is a unique way to decompose the perturbation
$f\in L^2(\R^n,\C^N)$
over the terms in \eqref{linearization-ansatz-1}
with factors $e^{\pm\jj\omega t}$:

\begin{lemma}\label{lemma-r-s}
Let $a,\,b\in\C$,
$\abs{a}^2-\abs{b}^2=1$,
and let \eqref{u-less-v}
be satisfied.
Then for any $\varrho\in L^2(\R^n,\C^N)$,
there is a unique choice of
scalar functions
$(P_j,Q_j)\sb{1\le j\le N/2}\in L^2(\R^n,\C)^N$
such that
%% the Ansatz
%% \[
%% \psi(t,x)
%% =
%% a\big(\phi_1(x)+\sum\sb{j=1}\sp{N/2}(P_j\varPhi_j(x)+Q_j\Chi_j(x))\big)
%% e^{-\jj\omega t}
%% +
%% b\big(\chi_1(x)+\sum\sb{j=1}\sp{N/2}(\bar Q_j\varPhi_j(x)+\bar P_j\Chi_j(x))\big)
%% e^{\jj\omega t}
%% \]
%% satisfies
%% \begin{eqnarray}\label{psi-zero-f}
%% \psi(t,x)\at{t=0}=a\phi_1(x)+b\chi_1(x)+\varrho(x).
%% \end{eqnarray}
\begin{eqnarray}\label{psi-zero-f}
a\sum\sb{j=1}\sp{N/2}\big(P_j\varPhi_j+Q_j\Chi_j\big)
+
b\sum\sb{j=1}\sp{N/2}\big(\bar Q_j\varPhi_j+\bar P_j\Chi_j\big)
=
\varrho.
\end{eqnarray}
The map
$
L^2(\R^n,\C^N)
\to L^2(\R^n,\C)^N,
$
$
\varrho\mapsto (P_j,Q_j)\sb{1\le j\le N/2}
$\,,
is continuous.
\end{lemma}

\begin{proof}
Let $\varrho\in L^2(\R^n,\C^N)$.
By Lemma~\ref{lemma-lin-independent},
there are
$f_j\in L^2(\R^n,\C)$
and
$g_j\in L^2(\R^n,\C)$
be such that
$
\varrho=\sum\sb{j=1}\sp{N/2}\varPhi_j f_j+\sum\sb{j=1}\sp{N/2}\Chi_j g_j,
$
and
the map
\[
L^2(\R^n,\C^N)\to L^2(\R^n,\C)^N,
\qquad
\varrho\mapsto (f_j,g_j)\sb{1\le j\le N/2}\in L^2(\R^n,\C)^N
\]
is continuous.
Equation \eqref{psi-zero-f} takes the form
\begin{eqnarray}\label{psi-zero-f-1}
a P_j(x)+b\bar Q_j(x)=f_j(x),
\qquad
a Q_j(x)+ b \bar P_j(x)=g_j(x),
\qquad
1\le j\le N/2.
\end{eqnarray}
Since $\abs{b/a}<1$,
for any $(f_j,g_j)\sb{1\le j\le N/2}\in L^2(\R^n,\C)^N$
the map
\[
(P_j,Q_j)
\mapsto
\Big(
\frac 1 a(f_j-b\bar Q_j),\,\frac 1 a(g_j-b\bar P_j)
\Big),
\qquad
1\le j\le N/2,
\]
is a contraction in $L^2(\R^n,\C)^N$
and thus has a unique fixed point
(a solution to \eqref{psi-zero-f-1})
which continuously depends on $(f_j,g_j)$.
\end{proof}

%% The above results could be summarized as follows:

%% \begin{lemma}
%% The linearization at a bi-frequency solitary wave
%% \eqref{bi-frequency-solitary-wave},
%% given by the system \eqref{rel-1}, \eqref{rel-2},
%% can be written in the form
%% \[
%% \jj\p_t p_k \phi_k
%% +\jj\p_t q_k \chi_k
%% -2\omega q_k \chi_k
%% =
%% (D_0 p_k)\phi_k
%% +(D_0 q_k)\chi_k
%% -2(v^2-u^2)f'(v^2-u^2)\beta\phi_1\Re p_1
%% \]
%% (the summation in $1\le k\le N/2$ is assumed),
%% with
%% $\big(p_j(t,x),q_j(t,x)\big)\sb{1\le j\le N/2}\in \C^N$.
%% \end{lemma}

We can now conclude the proof of Theorem~\ref{theorem-equiv}
in the case $n\le 4$, $N=4$.
By Proposition~\ref{prop-invariant-subspace}
and Lemma~\ref{lemma-r-s},
any solution to the linearization
at the bi-frequency solitary wave
\[
\theta(t,x)
=
a\phi_1(x) e^{-\jj\omega t}
+b\chi_1(x)e^{\jj\omega t},
\]
where
\[
\phi_1(x)
=
\begin{bmatrix}v(r)\xi\\\jj u(r)\xi\end{bmatrix},
\qquad
\chi_1(x)
=
\begin{bmatrix}-\jj u(r)\eta
\\
v(r)\eta
\end{bmatrix},
\]
\[
\xi,\,\eta\in\C^{N/2},
\quad\abs{\xi}=\abs{\eta}=1,
\qquad
a,\,b\in\C,
\quad
\abs{a}^2-\abs{b}^2=1
\]
(cf. \eqref{def-phi-chi-j})
can be written in the form
\begin{eqnarray}
\nonumber
\psi(t,x)
=
a\Big(
\phi_1(x)+\sum\sb{j=1}\sp{N/2}
\big(p_j(t,x)\phi_j(x)+q_j(t,x)\chi_j(x)\big)
\Big)e^{-\jj\omega t}
\\
\nonumber
+b\Big(\chi_1(x)
+\sum\sb{j=1}\sp{N/2}
\big(\bar q_j(t,x)\phi_j(x)+\bar p_j(t,x)\chi_j(x)\big)
\Big)e^{\jj\omega t};
\end{eqnarray}
this expression solves the nonlinear Dirac equation
in the zero and first order of the linear perturbation.
Taking in this Ansatz $a=1$ and $b=0$,
we see that
\[
\psi(t,x)
=
\Big(
\phi_1(x)+\sum\sb{j=1}\sp{N/2}
\big(p_j(t,x)\phi_j(x)+q_j(t,x)\chi_j(x)\big)
\Big)e^{-\jj\omega t}
\]
also solves the nonlinear Dirac equation
in the zero and first order of the perturbation.

We conclude that the bi-frequency solitary wave solution
\eqref{bsw}
%% $
%% \theta(t,x)=a\phi_{\xi}(x)e^{-\jj\omega t}
%% +
%% b\chi_{\eta}(x)e^{\jj\omega t}
%% $,
%% $\xi,\,\eta\in\C^{N/2}$,
%% $\abs{\xi}=\abs{\eta}=1$,
%% $a,\,b\in\C$,
%% $\abs{a}^2-\abs{b}^2=1$,
to the nonlinear Dirac equation \eqref{nld}
(or the Dirac--Klein--Gordon system \eqref{dkg})
is linearly unstable if and only if the one-frequency solitary wave
$\phi_{\xi}(x)e^{-\jj\omega t}$
is linearly unstable.
\end{proof}

%% \section{Acknowledgment}

%% The research of A. Comech was carried out
%% at the Institute for Information Transmission Problems
%% of the Russian Academy of Sciences
%% at the expense of the Russian Foundation
%% for Sciences (project 14-50-00150).

%%\bibliographystyle{alpha}
%\bibliographystyle{sima-doi}
%\bibliography{all}

\bibliographystyle{sima-doi}
\bibliography{bibcomech}

\providecommand{\etalchar}[1]{$^{#1}$}
\def\cprime{$'$} \def\polhk#1{\setbox0=\hbox{#1}{\ooalign{\hidewidth
  \lower1.5ex\hbox{`}\hidewidth\crcr\unhbox0}}} \def\cprime{$'$}
\begin{thebibliography}{CMKS{\etalchar{+}}16b}

\bibitem[AACDA13]{6646384}
A.~Aceves, A.~Auditore, M.~Conforti, and C.~De~Angelis,
  \href{http://dx.doi.org/10.1109/NLP.2013.6646384}{{\em Discrete localized
  modes in binary waveguide arrays\/}},
  \href{http://dx.doi.org/10.1109/NLP.2013.6646384}{in
  \href{http://dx.doi.org/10.1109/NLP.2013.6646384}{{\em Nonlinear Photonics
  ({NLP}), 2013 {IEEE} 2nd International Workshop\/}}}, pp. 38--42, 2013.

\bibitem[ACDAA13]{2013OptCo.297..125A}
A.~Auditore, M.~Conforti, C.~De~Angelis, and A.~B. Aceves,
  \href{http://dx.doi.org/10.1016/j.optcom.2013.01.068}{{\em Dark-antidark
  solitons in waveguide arrays with alternating positive-negative
  couplings\/}}, Optics Communications {\bf 297} (2013), pp. 125--128.

\bibitem[BC12a]{MR2892774}
G.~Berkolaiko and A.~Comech,
  \href{http://dx.doi.org/10.1051/mmnp/20127202}{{\em On spectral stability of
  solitary waves of nonlinear {D}irac equation in 1{D}\/}}, Math. Model. Nat.
  Phenom. {\bf 7} (2012), pp. 13--31.

\bibitem[BC12b]{MR2924465}
N.~Boussa{\"\i}d and S.~Cuccagna,
  \href{http://dx.doi.org/10.1080/03605302.2012.665973}{{\em On stability of
  standing waves of nonlinear {D}irac equations\/}}, Comm. Partial Differential
  Equations {\bf 37} (2012), pp. 1001--1056.

\bibitem[BC17a]{MR3670258}
N.~Boussa{\"\i}d and A.~Comech,
  \href{http://dx.doi.org/10.1137/16M1081385}{{\em Nonrelativistic asymptotics
  of solitary waves in the {D}irac equation with {S}oler-type nonlinearity\/}},
  SIAM J. Math. Anal. {\bf 49} (2017), pp. 2527--2572.

\bibitem[BC17b]{linear-b}
N.~Boussa{\"\i}d and A.~Comech, \href{http://arxiv.org/abs/1705.05481}{{\em
  Spectral stability of small amplitude solitary waves of the {D}irac equation
  with the {S}oler-type nonlinearity\/}}, ArXiv e-prints  (2017),
  \eprint{1705.05481}.

\bibitem[BSM{\etalchar{+}}06]{OL.31.001480}
A.~Betlej, S.~Suntsov, K.~G. Makris, L.~Jankovic, D.~N. Christodoulides, G.~I.
  Stegeman, J.~Fini, R.~T. Bise, and D.~J. DiGiovanni,
  \href{http://ol.osa.org/abstract.cfm?URI=ol-31-10-1480}{{\em All-optical
  switching and multifrequency generation in a dual-core photonic crystal
  fiber\/}}, Opt. Lett. {\bf 31} (2006), pp. 1480--1482.

\bibitem[CBKG02]{cerf2002security}
N.~J. Cerf, M.~Bourennane, A.~Karlsson, and N.~Gisin, {\em Security of quantum
  key distribution using d-level systems\/}, Physical Review Letters {\bf 88}
  (2002), p. 127902.

\bibitem[CG74]{MR0369952}
J.~M. Chadam and R.~T. Glassey,
  \href{http://dx.doi.org/10.1007/BF00250789}{{\em On certain global solutions
  of the {C}auchy problem for the (classical) coupled {K}lein-{G}ordon-{D}irac
  equations in one and three space dimensions\/}}, Arch. Rational Mech. Anal.
  {\bf 54} (1974), pp. 223--237.

\bibitem[CGG14]{MR3208458}
A.~Comech, M.~Guan, and S.~Gustafson,
  \href{http://dx.doi.org/10.1016/j.anihpc.2013.06.001}{{\em On linear
  instability of solitary waves for the nonlinear {D}irac equation\/}}, Ann.
  Inst. H. Poincar\'e Anal. Non Lin\'eaire {\bf 31} (2014), pp. 639--654.

\bibitem[CMKS{\etalchar{+}}16a]{PhysRevLett.116.214101}
J.~Cuevas-Maraver, P.~G. Kevrekidis, A.~Saxena, A.~Comech, and R.~Lan,
  \href{http://link.aps.org/doi/10.1103/PhysRevLett.116.214101}{{\em Stability
  of solitary waves and vortices in a 2{D} nonlinear {D}irac model\/}}, Phys.
  Rev. Lett. {\bf 116} (2016), p. 214101.

\bibitem[CMKS{\etalchar{+}}16b]{JSTQE.2015.2485607}
J.~Cuevas-Maraver, P.~G. Kevrekidis, A.~Saxena, F.~Cooper, A.~Khare, A.~Comech,
  and C.~M. Bender, \href{http://dx.doi.org/10.1109/JSTQE.2015.2485607}{{\em
  Solitary waves of a $\mathcal{PT}$-symmetric nonlinear {D}irac equation\/}},
  IEEE Journal of Selected Topics in Quantum Electronics {\bf 22} (2016), pp.
  1--9.

\bibitem[CPS17]{MR3592683}
A.~{Comech}, T.~V. {Phan}, and A.~{Stefanov},
  \href{http://dx.doi.org/10.1016/j.anihpc.2015.11.001}{{\em {Asymptotic
  stability of solitary waves in generalized {G}ross--{N}eveu model}\/}}, Ann.
  Inst. H. Poincar\'e Anal. Non Lin\'eaire {\bf 34} (2017), pp. 157--196.

\bibitem[DKCK04]{PhysRevA.69.032313}
T.~Durt, D.~Kaszlikowski, J.-L. Chen, and L.~C. Kwek,
  \href{https://link.aps.org/doi/10.1103/PhysRevA.69.032313}{{\em Security of
  quantum key distributions with entangled qudits\/}}, Phys. Rev. A {\bf 69}
  (2004), p. 032313.

\bibitem[EdSS97]{JOSAB.14.002980}
B.~J. Eggleton, C.~M. de~Sterke, and R.~E. Slusher,
  \href{http://josab.osa.org/abstract.cfm?URI=josab-14-11-2980}{{\em Nonlinear
  pulse propagation in {B}ragg gratings\/}}, J. Opt. Soc. Am. B {\bf 14}
  (1997), pp. 2980--2993.

\bibitem[Gal77]{MR0503135}
A.~Galindo, \href{https://doi.org/10.1007/BF02785129}{{\em A remarkable
  invariance of classical {D}irac {L}agrangians\/}}, Lett. Nuovo Cimento (2)
  {\bf 20} (1977), pp. 210--212.

\bibitem[GN74]{PhysRevD.10.3235}
D.~J. Gross and A.~Neveu,
  \href{http://dx.doi.org/10.1103/PhysRevD.10.3235}{{\em Dynamical symmetry
  breaking in asymptotically free field theories\/}}, Phys. Rev. D {\bf 10}
  (1974), pp. 3235--3253.

\bibitem[Iva38]{jetp.8.260-short}
D.~D. Ivanenko, {\em Notes to the theory of interaction via particles\/}, Zh.
  \'Eksp. Teor. Fiz {\bf 8} (1938), pp. 260--266.

\bibitem[KRR{\etalchar{+}}17]{kues2017}
M.~Kues, C.~Reimer, P.~Roztocki, L.~R. CortИs, S.~Sciara, B.~Wetzel, Y.~Zhang,
  A.~Cino, S.~T. Chu, B.~E. Little, D.~J. Moss, L.~Caspani, J.~AzaЯa, and
  R.~Morandotti, \href{http://dx.doi.org/10.1038/nature22986}{{\em On-chip
  generation of high-dimensional entangled quantum states and their coherent
  control\/}}, Nature {\bf 546} (2017), pp. 622--626.

\bibitem[LG75]{PhysRevD.12.3880}
S.~Y. Lee and A.~Gavrielides,
  \href{http://dx.doi.org/10.1103/PhysRevD.12.3880}{{\em Quantization of the
  localized solutions in two-dimensional field theories of massive
  fermions\/}}, Phys. Rev. D {\bf 12} (1975), pp. 3880--3886.

\bibitem[LT13]{PhysRevLett.110.053901}
N.~Lazarides and G.~P. Tsironis,
  \href{http://link.aps.org/doi/10.1103/PhysRevLett.110.053901}{{\em
  Gain-driven discrete breathers in $\mathcal{P}\mathcal{T}$-symmetric
  nonlinear metamaterials\/}}, Phys. Rev. Lett. {\bf 110} (2013), p. 053901.

\bibitem[MF16]{melnikov2016}
A.~A. Melnikov and L.~E. Fedichkin,
  \href{http://dx.doi.org/10.1038/srep34226}{{\em Quantum walks of interacting
  fermions on a cycle graph\/}}, Sci. Rep. {\bf 6} (2016), p. 34226.

\bibitem[MLB14]{PhysRevLett.113.150401}
A.~Marini, S.~Longhi, and F.~Biancalana,
  \href{http://link.aps.org/doi/10.1103/PhysRevLett.113.150401}{{\em Optical
  simulation of neutrino oscillations in binary waveguide arrays\/}}, Phys.
  Rev. Lett. {\bf 113} (2014), p. 150401.

\bibitem[MMS{\etalchar{+}}04]{OL.29.002890}
R.~Morandotti, D.~Mandelik, Y.~Silberberg, J.~S. Aitchison, M.~Sorel, D.~N.
  Christodoulides, A.~A. Sukhorukov, and Y.~S. Kivshar,
  \href{http://ol.osa.org/abstract.cfm?URI=ol-29-24-2890}{{\em Observation of
  discrete gap solitons in binary waveguide arrays\/}}, Opt. Lett. {\bf 29}
  (2004), pp. 2890--2892.

\bibitem[OY04]{MR2082456}
T.~Ozawa and K.~Yamauchi,
  \href{http://projecteuclid.org/euclid.die/1356060310}{{\em Structure of
  {D}irac matrices and invariants for nonlinear {D}irac equations\/}},
  Differential Integral Equations {\bf 17} (2004), pp. 971--982.

\bibitem[PS12]{MR2985264}
D.~E. Pelinovsky and A.~Stefanov,
  \href{http://dx.doi.org/10.1063/1.4731477}{{\em Asymptotic stability of small
  gap solitons in nonlinear {D}irac equations\/}}, J. Math. Phys. {\bf 53}
  (2012), pp. 073705, 27.

\bibitem[SLL{\etalchar{+}}12]{1751-8121-45-44-444029}
J.~Schindler, Z.~Lin, J.~M. Lee, H.~Ramezani, F.~M. Ellis, and T.~Kottos,
  \href{http://stacks.iop.org/1751-8121/45/i=44/a=444029}{{\em
  $\mathcal{PT}$-symmetric electronics\/}}, Journal of Physics A: Mathematical
  and Theoretical {\bf 45} (2012), p. 444029.

\bibitem[SLZ{\etalchar{+}}11]{PhysRevA.84.040101}
J.~Schindler, A.~Li, M.~C. Zheng, F.~M. Ellis, and T.~Kottos,
  \href{http://link.aps.org/doi/10.1103/PhysRevA.84.040101}{{\em Experimental
  study of active \textit{LRC} circuits with $\mathcal{PT}$ symmetries\/}},
  Phys. Rev. A {\bf 84} (2011), p. 040101.

\bibitem[Sol70]{PhysRevD.1.2766}
M.~Soler, \href{http://dx.doi.org/10.1103/PhysRevD.1.2766}{{\em Classical,
  stable, nonlinear spinor field with positive rest energy\/}}, Phys. Rev. D
  {\bf 1} (1970), pp. 2766--2769.

\bibitem[Thi58]{MR0091788}
W.~E. Thirring, \href{http://dx.doi.org/10.1016/0003-4916(58)90015-0}{{\em A
  soluble relativistic field theory\/}}, Ann. Physics {\bf 3} (1958), pp.
  91--112.

\bibitem[Wak66]{wakano-1966}
M.~Wakano, \href{http://ptp.ipap.jp/link?PTP/35/1117/}{{\em Intensely localized
  solutions of the classical {D}irac-{M}axwell field equations\/}}, Progr.
  Theoret. Phys. {\bf 35} (1966), pp. 1117--1141.

\end{thebibliography}
%% To add your bib data:
%% create bib<yourname>.bib
%% add your entry to the \bibliography command
%% svn add bib<yourname>.bib
%% svn commit -m YOURNAME

\end{document}